\documentclass[12pt,a4paper,UKenglish]{article}
\usepackage[margin=1in]{geometry}

\usepackage{amssymb}
\usepackage{latexsym}
\usepackage{amsmath, amsthm, mathtools}
\usepackage{bussproofs}
\usepackage[shortlabels]{enumitem} %

\usepackage{hyperref}

\usepackage{multicol}
\usepackage{tikz}

\usetikzlibrary{snakes}

\usepackage{mathtools}

\usepackage{authblk}

\usepackage{todonotes}

\newtheorem{theorem}{Theorem}

\newtheorem{lemma}[theorem]{Lemma}
\newtheorem{proposition}[theorem]{Proposition}
\theoremstyle{remark}
\newtheorem{remark}[theorem]{Remark}

\newtheorem*{switching lemma}{Switching Lemma}
\newtheorem*{open problem}{Open Problem}
\theoremstyle{definition}
\newtheorem{definition}[theorem]{Definition}
\newtheorem*{definition*}{Definition}

\DeclareMathOperator{\dom}{dom}
\DeclareMathOperator{\im}{im}

\title{Resolution Lower Bounds for Refutation Statements} %

\author{Michal Garl\'{\i}k\thanks{Funded by European Research Council (ERC) under the European Union's Horizon 2020 research and innovation programme, grant agreement ERC-2014-CoG 648276 (AUTAR).}
}
\affil{Dept. Ci\`encies de la Computaci\'o\\
Universitat Polit\`ecnica de Catalunya\\
C. Jordi Girona, 1-3\\
08034 Barcelona, Spain\\
email: \texttt{mgarlik@cs.upc.edu}
}

\begin{document}

\maketitle

\begin{abstract}
For any unsatisfiable CNF formula we give an exponential lower bound on the size of resolution refutations of a propositional statement that the formula has a resolution refutation. We describe three applications. (1) An open question in \cite{atserias-muller2019-ref} asks whether a certain natural propositional encoding of the above statement is hard for Resolution. We answer by giving an exponential size lower bound. (2) We show exponential resolution size lower bounds for reflection principles, thereby improving a result in \cite{Atserias-Bonet2004}. (3) We provide new examples of CNFs that exponentially separate Res(2) from Resolution (an exponential separation of these two proof systems was originally proved in \cite{segerlind-buss-impagliazzo}).
 \end{abstract}

\section{Introduction}
Proving lower bounds on the size of propositional proofs is the central task of proof complexity theory. After Cook and Reckhow \cite{cook-reckhow1979} motivated this line of research as an approach towards establishing NP $\neq$ coNP, some initial success for weak proof systems followed, e.g., the first exponential size lower bound for Resolution was proved by Haken \cite{Haken1985}. Nevertheless, many important open problems from the 1980s and 1990s remain unsolved, and it seems that proving nontrivial lower bounds on the size of propositional proofs is hard. If it is hard for people, it is natural to ask if it is also hard for the proof systems themselves. In trying to formalize this question so that it makes sense to a proof system, we must say what we mean by `proving is hard'. It can be `there are no short proofs', a statement which occurs as a part of reflection principles. By `short' we mean polynomial in the size of the formula being proven or refuted. The negation of the \emph{reflection principle} for a proof system $P$ is a conjunction of the statement `$y$ is a $P$-refutation of length $s$ of formula $x$ of length $n$' and the statement `$z$ is a satisfying assignment of formula $x$'. In a propositional formulation of the principle, $P,s,n$ are fixed parameters and $x,y,z$ are disjoint sets of variables. 
A possible way to formalize the above question is then to take the first conjunct of the negation of the reflection principle and plug in  for the $x$-variables some formula $F$ of length $n$. The resulting formula was discussed and utilized by Pudl\'ak \cite{Pudlak2003}; we denote it by $\textnormal{REF}^F_{P,s}$ and call it a \emph{refutation statement} for $P$. We may now ask whether some proof system $Q$ can shortly refute $\textnormal{REF}^F_{P,s}$, and if it can not, we can interpret this to mean that lower bounds for $P$-refutations of $F$ are hard for $Q$. 

Pudl\'ak \cite{Pudlak2003} found connections between the reflection principles and automatizability, and these were elaborated on in \cite{Atserias-Bonet2004}. Following \cite{Bonet-Pitassi-Raz}, a proof system $P$ is \emph{automatizable} if there is a deterministic algorithm that when given as input an unsatisfiable CNF formula $F$ outputs its $P$-refutation in time polynomial in the size of the shortest $P$-refutation of $F$. Recently, Atserias and M\"uller \cite{atserias-muller2019-ref} showed that Resolution is not automatizable unless $\textnormal{P} = \textnormal{NP}$. Refutation statements for Resolution play a prominent role in their proof. They show that strong enough resolution size lower bounds for $\textnormal{REF}^F_{\textnormal{Res},s}$ with an unsatisfiable $F$ imply their result. However, they leave the lower bound problem for $\textnormal{REF}^F_{\textnormal{Res},s}$ as an open question, and in place of $\textnormal{REF}^F_{\textnormal{Res},s}$ they use in the proof a different formulation of the refutation statement, obtained by a \emph{relativization} of $\textnormal{REF}^F_{\textnormal{Res},s}$, for which lower bounds are easier to get. In this paper we focus mainly on giving an answer to the question.

\subsection{Results in This Paper}
The result that requires the most work is the following lower bound.
\begin{theorem} \label{thm:main-lower-bound-REF-F-s-t-Introduction}
For each $\epsilon>0$ there is $\delta>0$ and an integer $t_0$ such that if $n,r,s,t$ are integers satisfying $t\geq s \geq n+1$, $r \geq n \geq 2$, $t \geq r^{3+\epsilon}$, $t \geq t_0$, and $F$ is an unsatisfiable CNF consisting of $r$ clauses $C_1, \ldots , C_r$ in $n$ variables $x_1 , \ldots , x_n$, then any resolution refutation of $\textnormal{REF}^{F}_{s,t}$ has length greater than $2^{t^{\delta}}$.
\end{theorem}
Here $\textnormal{REF}^{F}_{s,t}$, missing the lower index denoting the proof system as we concentrate on resolution refutation statements, is a variant of the refutation statement insisting that the resolution refutation it describes has the form of a levelled graph. 
A similar simplifying assumption, making it more practical to design random restrictions, is used in \cite{Thapen2016} for a propositional version of the coloured polynomial local search principle. Our proof proceeds with defining a random restriction tailored to $\textnormal{REF}^{F}_{s,t}$ and to an adversary argument. The nature of the refutation statement and the fact that the relations between refutation lines are encoded in unary, rather than in binary, necessitate a more complicated adversary argument than in  \cite{Pudlak2000} or \cite{Thapen2016}, and this in turn poses more requirements on the random restriction. We discuss these details after the proof, in Remarks \ref{remark:bin-encoding} and \ref{remark:obstacles}. 

We then show that Theorem \ref{thm:main-lower-bound-REF-F-s-t-Introduction} implies an exponential resolution size lower bound for the encoding of the refutation statement for which the lower bound question in \cite{atserias-muller2019-ref} is originally asked.

We see two reasons for working with the unary encoding of $\textnormal{REF}^{F}_{s,t}$. First, $\textnormal{REF}^{F}_{s,t}$ is weaker than refutation statements encoded in binary or relativized refutation statements. Hence lower bounds for $\textnormal{REF}^{F}_{s,t}$ imply lower bounds for the other encodings. Second, researchers who dealt with propositional encodings of reflection principles or refutation statements opted for the unary encoding \cite{ Atserias-Bonet2004, Krajicek-Skelley-Thapen, Pudlak2003}.

Our next result is that the negation of the reflection principle for Resolution, expressed by the formula $\textnormal{SAT}^{n,r} \land \textnormal{REF}^{n,r}_{s,t}$, exponentially separates the system Res(2) from Resolution. It was shown by Atserias and Bonet \cite{Atserias-Bonet2004} that a similar encoding of the negation of the reflection principle separates the two theories almost-exponentially (giving a $2^{\Omega(2^{\log^{\epsilon} n})}$ resolution lower bound and a polynomial $\textnormal{Res}(2)$ upper bound). The exponential separation of $\textnormal{Res}(2)$ from Resolution was originally proved in \cite{segerlind-buss-impagliazzo} using a variation of the graph ordering principle. 
Our lower bound is stated in Theorem \ref{thm:refl-princ-Lower-bound-Introduction} below. 
\begin{theorem}
\label{thm:refl-princ-Lower-bound-Introduction}
For every $c>4$ there is $\delta >0$ and an integer $n_0$ such that if $n,r,s,t$ are integers satisfying $t\geq s \geq n+1$, $r \geq n \geq n_0$, $n^c \geq t \geq r^4$, then any resolution refutation of $\textnormal{SAT}^{n,r} \land \textnormal{REF}^{n,r}_{s,t}$ has length greater than $2^{n^{\delta}}$.
\end{theorem}
The proof of the theorem also yields new examples of CNFs exponentially separating $\textnormal{Res}(2)$ from Resolution.

\begin{theorem}
\label{thm:examples-separating-Introduction}
Let $\delta_1 > 0$ and let $\{A_n\}_{n \geq 1}$ be a family of unsatisfiable CNFs such that $A_n$ is in $n$ variables, has the number of clauses polynomial in $n$, and has no resolution refutations of length at most $2^{n^{\delta_1}}$. Then there is $\delta > 0$ and a polynomial $p$ such that $A_n \land \textnormal{REF}^{A_n}_{n+1,p(n)}$ has no resolution refutations of length at most $2^{n^{\delta}}$ and has polynomial size $\textnormal{Res}(2)$ refutations.
\end{theorem}

A $\textnormal{Res}(2)$ upper bound for $\textnormal{SAT}^{n,r} \land \textnormal{REF}^{n,r}_{s,t}$, needed for completing the separation by this formula as well as by the formulas in Theorem \ref{thm:examples-separating-Introduction}, is stated in the following theorem.

\begin{theorem}
\label{thm:refl-princ-Upper-bound-Introduction}
The negation of the reflection principle for Resolution expressed by the formula $\textnormal{SAT}^{n,r} \land \textnormal{REF}^{n,r}_{s,t}$ has $\text{Res}(2)$ refutations of size $O(trn^2+ tr^2 + st^2n^3 + st^3n)$.
\end{theorem}

A polynomial size $\textnormal{Res}(2)$ upper bound on a similar encoding of the negation of the reflection principle for Resolution was proved in \cite{Atserias-Bonet2004}. We simplify the proof and adapt it to $\textnormal{SAT}^{n,r} \land \textnormal{REF}^{n,r}_{s,t}$.

\subsection{Outline of the Paper}
The rest of the paper is organized as follows.

In Section \ref{sec:preliminaries} we give the necessary preliminaries. 

In Section \ref{sec:res-refut-of-s-t-definition}, Resolution of $s$ levels of $t$ clauses is introduced, and the clauses of the refutation statement $\textnormal{REF}^{F}_{s,t}$ for this refutation system are listed. We show that this system simulates Resolution with at most quadratic increase in length.

In Section \ref{sec:a-lower-bound-on-refFst} we prove Theorem \ref{thm:main-lower-bound-REF-F-s-t-Introduction}. That Theorem \ref{thm:main-lower-bound-REF-F-s-t-Introduction} also answers the original lower bound question from \cite{atserias-muller2019-ref} is shown in Appendix \ref{sec:REF-AM}.

In Section \ref{sec:refl-princ-for-res} we define the formula $\textnormal{SAT}^{n,r} \land \textnormal{REF}^{n,r}_{s,t}$ and we prove Theorems \ref{thm:refl-princ-Lower-bound-Introduction} and \ref{thm:examples-separating-Introduction}. 

In Section \ref{sec:refl-princ-upper-bound} we prove Theorem \ref{thm:refl-princ-Upper-bound-Introduction}.

\section{Preliminaries}
\label{sec:preliminaries}
For an integer $s$, the set $\{1, \ldots, s\}$ is denoted by $[s]$. 
We write $\dom(\sigma), \im(\sigma)$ for the domain and image of a function $\sigma$.
Two functions $\sigma, \tau$ are \emph{compatible} if $\sigma \cup \tau$ is a function.
If $x$ is a propositional variable, the \emph{positive literal of} $x$, denoted by $x^1$, is $x$, and the \emph{negative literal of} $x$, denoted by $x^0$, is $\neg x$. A \emph{clause} is a set of literals. A clause is written as a disjunction of its elements. A \emph{term} is a set of literals, and is written as a conjunction of the literals. A \emph{CNF} is a set of clauses, written as a conjunction of the clauses. A $k$-\emph{CNF} is a CNF whose every clause has at most $k$ literals. A \emph{DNF} is a set of terms, written as a disjunction of the terms. A $k$-\emph{DNF} is a DNF whose every term has at most $k$ literals. We will identify 1-DNFs with clauses. 
A clause is \emph{non-tautological} if it does not contain both the positive and negative literal of the same variable. A clause $C$ is a \emph{weakening} of a clause $D$ if $D \subseteq C$. A clause $D$ is the \emph{resolvent of} clauses $C_1$ and $C_2$ \emph{on} a variable $x$ if $x \in C_1, \neg x \in C_2$ and $D = (C_1\setminus \{x\}) \cup (C_2 \setminus \{\neg x\})$. If $E$ is a weakening of the resolvent of $C_1$ and $C_2$ on $x$, we say that $E$ is obtained by the \emph{resolution rule} from $C_1$ and $C_2$, and we call $C_1$ and $C_2$ the \emph{premises} of the rule.

Let $F$ be a CNF and $C$ a clause. A \emph{resolution derivation of} $C$ \emph{from} $F$ is a sequence of clauses $\Pi = (C_1, \ldots , C_s)$ such that $C_s = C$ and for all $u \in [s]$, $C_u$ is a weakening of a clause in $F$, or there are $v,w \in [u-1]$ such that $C_u$ is obtained by the resolution rule from $C_v$ and $C_w$. 
A \emph{resolution refutation of} $F$ is a resolution derivation of the empty clause from $F$.
The \emph{length} of a resolution derivation $\Pi = (C_1, \ldots , C_s)$ is $s$. 
For $u \in [s]$, the \emph{height of} $u$ \emph{in} $\Pi$ is the maximum $h$ such that there is a subsequence $(C_{u_1}, \ldots , C_{u_h})$ of $\Pi$ in which $u_h = u$ and for each $i \in [h-1]$, $C_{u_i}$ is a premise of a resolution rule by which $C_{u_{i+1}}$ is obtained in $\Pi$. 
The \emph{height} of $\Pi$ is the maximum height of $u$ in $\Pi$ for $u \in [s]$. 

A \emph{partial assignment} to the variables $x_1, \ldots , x_n$ is a partial map from $\{ x_1, \ldots , x_n \}$ to $\{0,1\}$. Let $\sigma$ be a partial assignment. The CNF $F \! \restriction \! \sigma$ is formed from $F$ by removing every clause containing a literal satisfied by $\sigma$, and removing every literal falsified by $\sigma$ from the remaining clauses. If $\Pi = (C_1, \ldots , C_s)$ is a sequence of clauses, $\Pi \!\restriction \! \sigma$ is formed from $\Pi$ by the same operations. Note that if $\Pi$ is a resolution refutation of $F$, then $\Pi \!\restriction \! \sigma$ is a resolution refutation of $F \!\restriction \! \sigma$. 

The $\text{Res}(k)$ refutation system is a generalization of Resolution. Its lines are $k$-DNFs and it has the following inference rules ($A,B$ are $k$-DNFs, $j \in [k]$, and $l,l_1,\ldots, l_j$ are literals):
\begingroup
\setlength{\tabcolsep}{15pt} %
\renewcommand{\arraystretch}{2.5}  %
\begin{center}
\begin{tabular}{c c}
\AxiomC{$A \lor l_1$}
\AxiomC{$B \lor (l_2 \land \cdots \land l_j)$}
\RightLabel{\,$\land$-introduction}
\BinaryInfC{$A \lor B \lor (l_1 \land \cdots \land l_j)$}
\DisplayProof 
&
\AxiomC{\phantom{A}}
\RightLabel{\,Axiom}
\UnaryInfC{$x \lor \neg x$}
\DisplayProof 
\\
\AxiomC{$A \lor (l_1 \land \cdots \land l_j)$}
\AxiomC{$B \lor \neg l_1 \lor \cdots \lor \neg l_j$}
\RightLabel{\,Cut}
\BinaryInfC{$A \lor B$}
\DisplayProof 
& 
\AxiomC{$A$}
\RightLabel{\,Weakening}
\UnaryInfC{$A \lor B$}
\DisplayProof
\\
\end{tabular}
\end{center}
\endgroup
\noindent Let $F$ be a CNF. A $\textit{Res}(k)$ \emph{derivation from} $F$ is a sequence of $k$-DNFs ($D_1,\ldots, D_s$) so that each $D_i$ either belongs to $F$ or follows from the preceding lines by an application of one of the inference rules. A $\textit{Res}(k)$ \emph{refutation of} $F$ is a $\text{Res}(k)$ derivation from $F$ whose final line is the empty clause. The \emph{size} of a $\text{Res}(k)$ derivation is the number of symbols in it.

\section{Resolution Refutations of $\textit{s}$ Levels of $\textit{t}$ Clauses}
\label{sec:res-refut-of-s-t-definition}
We introduce a variant of Resolution in which the clauses forming a refutation are arranged in layers. 
\begin{definition}
 Let $F$ be a CNF of $r$ clauses in $n$ variables $x_1, \ldots, x_n$.  We say that $F$ has a \emph{resolution refutation of $s$ levels of $t$ clauses} if there is a sequence of clauses $C_{i,j}$ indexed by all pairs $(i,j) \in [s] \times [t]$, such that each clause $C_{1,j}$ on the first level is a weakening of a clause in $F$, each clause $C_{i,j}$ on level $i \! \in \! [s] \! \setminus \! \{1\}$ is a weakening of the resolvent of two clauses from level $i-1$ on a variable, and the clause $C_{s,t}$ is empty. 
\end{definition}

The following proposition shows that this system quadratically simulates Resolution and preserves the refutation height.
The proof uses a simple self-replicating pattern both to transport a premise of the resolution rule to the required level and to fill in all clauses $C_{i,j}$ that do not directly participate in the simulation.

\begin{proposition}
\label{proposition:simulation-of-Res}
If a $(n-1)$-CNF $F$ in $n$ variables has a resolution refutation of height $h$ and length $s$, then $F$ has a resolution refutation of $h$ levels of $3s$ clauses. 
\end{proposition}

\begin{proof}
Let $\Pi$ be a resolution refutation of $F$ of height $h$ and length $s$. Assume that $\Pi$ is $(C_1, \ldots , C_s)$, and that without loss of generality $C_j$, $j \in [s]$, is non-tautological and either $C_j \in F$ or $C_j$ is the resolvent of $C_{j_1}, C_{j_2}$ for some $j_1,j_2 < j$. For $j \in [s]$, let $h_j$ be the height of $j$ in $\Pi$. We prove by induction on $s' \in [s]$ the following: There is a resolution derivation of $C_{s'}$ of $h$ levels of $3s'$ clauses such that for each $j \in [s']$ and $i \in \{h_j,\ldots,h\}$, $C_{i, 3j} = C_j$.

Base step: $s'=1$. Pick a variable $x$ such that no literal of $x$ is in $C_1$. Such a variable exists by our assumptions. For all $i \in [h]$, set $(C_{i,1},C_{i,2},C_{i,3}) = (C_1 \cup \{x\}, C_1 \cup \{\neg x\}, C_1)$. This is a valid resolution derivation of $h$ levels of 3 clauses: the clauses on the first level are weakenings of $C_1 \in F$, and each clause on any subsequent level is derived from the first and second clause of the previous level.

Induction step: Assume the statement holds for $s'$, as witnessed by a derivation $\Pi'$. We prove it for $s'+1$. Pick a variable $x$ such that no literal of $x$ is in $C_{s'+1}$; it exists by our assumptions. 
If $C_{s'+1} \in F$, define, for each $i\in [h]$,  $(C_{i,3s'+1},C_{i,3s'+2},C_{i,3s'+3}) = (C_{s'+1} \cup \{x\}, C_{s'+1} \cup \{\neg x\}, C_{s'+1})$. 
If $C_{s'+1}$ is obtained by the resolution rule in $\Pi$, the premises of the rule appear on level $h_{s'+1}-1$ in $\Pi'$ by the induction hypothesis. So we can extend $\Pi'$ by defining, for $i \in \{h_{s'+1},\ldots , h\}$, $(C_{i,3s'+1},C_{i,3s'+2},C_{i,3s'+3}) = (C_{s'+1} \cup \{x\}, C_{s'+1} \cup \{\neg x\}, C_{s'+1})$. Next, for each $i \in [h_{s'+1}-1]$, define $(C_{i,3s'+1},C_{i,3s'+2},C_{i,3s'+3}) = (C_1 \cup \{x\}, C_1 \cup \{\neg x\}, C_1)$. It is easy to check that this is a valid derivation and it satisfies the required properties.
\end{proof}

We proceed to our formalization of the refutation statement for this refutation system. 
Let $n,r,s,t$ be integers. Let $F$ be a CNF consisting of $r$ clauses $C_1, \ldots , C_r$ in $n$ variables $x_1, \ldots , x_n$. We define a propositional formula $\textnormal{REF}^{F}_{s,t}$ expressing that $F$ has a resolution refutation of $s$ levels of $t$ clauses.

We first list the variables of $\textnormal{REF}^{F}_{s,t}$. $D$-\emph{variables} $D(i,j,\ell,b)$, $i \in [s]$, $j \in [t], \ell \in [n], b \in \{0,1\}$, encode clauses $C_{i,j}$ as follows: $D(i,j,\ell,1)$ (resp. $D(i,j,\ell,0)$) means that the literal $x_\ell$ (resp. $\neg x_\ell$) is in $C_{i,j}$.  
$L$-\emph{variables} $L(i,j,j')$ (resp. $R$-\emph{variables} $R(i,j,j')$), $i \! \in \! [s] \! \setminus \! \{1\}, j,j' \in [t]$, say that $C_{i-1,j'}$ is a premise of the resolution rule by which $C_{i,j}$ is obtained, and it is the premise containing the positive (resp. negative) literal of the resolved variable. 
$V$-\emph{variables} $V(i,j,\ell)$, $i \! \in \! [s] \! \setminus \! \{1\}, j \in [t], \ell \in [n]$, say that $C_{i,j}$ is obtained by resolving on $x_\ell$.
$I$-\emph{variables} $I(j,m)$, $j \in [t], m \in [r]$, say that $C_{1,j}$ is a weakening of $C_m$.

$\textnormal{REF}^{F}_{s,t}$ is the union of the following fifteen sets of clauses:
\begin{alignat}{2}
\label{refclause:axioms} & \neg I(j,m) \lor D(1,j,\ell,b)   & j \! \in \! [t], m \! \in \! [r], b \! \in \! \{0,1\}, x_\ell^b \! \in \! C_m,
\intertext{clause $C_{1,j}$  contains the literals of $C_m$ assigned to it by $I(j,m)$,}
\label{refclause:nontaut} & \neg D(i,j,\ell,1) \lor \neg D(i,j,\ell,0)  & i \! \in \! [s], j \! \in \! [t], \ell \! \in \! [n],
\intertext{no clause $C_{i,j}$ contains $x_\ell$ and $\neg x_\ell$ at the same time,}
\label{refclause:res-L-cut} & \neg L(i,j,j') \lor \neg V(i,j,\ell) \lor D(i-1,j',\ell,1) & i \! \in \! [s] \! \setminus \! \{1\}, j,j' \! \in \! [t], \ell \! \in \! [n], 
\\
\label{refclause:res-R-cut} & \neg R(i,j,j') \lor \neg V(i,j,\ell) \lor D(i-1,j',\ell,0) & i \! \in \! [s] \! \setminus \! \{1\}, j,j' \! \in \! [t], \ell \! \in \! [n],
\intertext{clause $C_{i-1,j'}$ used as the premise given by $L(i,j,j')$ (resp. $R(i,j,j')$) in resolving on $x_\ell$ must contain $x_\ell$ (resp. $\neg x_\ell$),}
\begin{split}
\label{refclause:res-L-transf} \mathrlap{ \neg L(i,j,j') \lor \neg V(i,j,\ell) \lor \neg D(i-1,j',\ell',b) \lor D(i,j,\ell',b)} \\
\mathrlap{ \, \, \; \; \; \qquad \qquad \qquad \qquad  i \! \in \! [s] \! \setminus \! \{1\}, j,j' \! \in \! [t], \ell,\ell' \! \in \! [n], b \! \in \! \{0,1\}, (\ell',b) \neq (\ell,1),} 
\end{split}
\\[1ex]
\begin{split}
\label{refclause:res-R-transf} \mathrlap{ \neg R(i,j,j') \lor \neg V(i,j,\ell) \lor \neg D(i-1,j',\ell',b) \lor D(i,j,\ell',b)} \\
 \mathrlap{ \, \, \; \; \; \qquad \qquad \qquad \qquad  i \! \in \! [s] \! \setminus \! \{1\}, j,j' \! \in \! [t], \ell,\ell' \! \in \! [n], b \! \in \! \{0,1\}, (\ell',b) \neq (\ell,0),}
\end{split}
\intertext{clause $C_{i,j}$ derived by resolving on $x_\ell$ must contain each literal different from $x_\ell$ (resp. $\neg x_\ell$) from the premise given by $L(i,j,j')$ (resp. $R(i,j,j')$),} 
\label{refclause:empty-clause} & \neg D(s,t,\ell,b) & \ell \! \in \! [n], b \! \in \! \{0,1\}, 
\intertext{clause $C_{s,t}$ is empty,}
\label{refclause:V-dom} & V(i,j,1) \lor V(i,j,2) \lor \ldots \lor V(i,j,n)  & i \! \in \! [s] \! \setminus \! \{1\}, j\! \in \! [t], \\
\label{refclause:I-dom} & I(j,1) \lor I(j,2) \lor \ldots \lor I(j,r)  & j\! \in \! [t], \\
\label{refclause:L-dom} & L(i,j,1) \lor L(i,j,2) \lor \ldots \lor L(i,j,t)  & i \! \in \! [s] \! \setminus \! \{1\}, j\! \in \! [t], \\
\label{refclause:R-dom} & R(i,j,1) \lor R(i,j,2) \lor \ldots \lor R(i,j,t)  & i \! \in \! [s] \! \setminus \! \{1\}, j\! \in \! [t], \\
\label{refclause:V-func} & \neg V(i,j,\ell) \lor \neg V(i,j,\ell')  & i \! \in \! [s] \! \setminus \! \{1\}, j\! \in \! [t], \ell,\ell' \! \in \! [n], \ell \neq \ell',\\
\label{refclause:I-func} & \neg I(j,m) \lor \neg I(j,m')  &  j\! \in \! [t], m,m' \! \in \! [r], m \neq m',\\
\label{refclause:L-func} & \neg L(i,j,j') \lor \neg L(i,j,j'')  & i \! \in \! [s] \! \setminus \! \{1\}, j,j',j'' \! \in \! [t], j' \neq j'',\\
\label{refclause:R-func} & \neg R(i,j,j') \lor \neg R(i,j,j'')  & i \! \in \! [s] \! \setminus \! \{1\}, j,j',j'' \! \in \! [t], j' \neq j'', 
\end{alignat}
the $V,I,L,R$-variables define functions with the required domains and ranges.

\section{A Lower Bound on Lengths of Resolution Refutations of Resolution Refutation Statements}
\label{sec:a-lower-bound-on-refFst}

We restate Theorem \ref{thm:main-lower-bound-REF-F-s-t-Introduction} from the Introduction.

\begin{theorem} \label{thm:main-lower-bound-REF-F-s-t}
For each $\epsilon>0$ there is $\delta>0$ and an integer $t_0$ such that if $n,r,s,t$ are integers satisfying 
\begin{equation} \label{eqn:conditions-on-nrst}
 t\geq s \geq n+1,\qquad r \geq n \geq 2,\qquad t \geq r^{3+\epsilon}, \qquad t \geq t_0,
\end{equation}
and $F$ is an unsatisfiable CNF consisting of $r$ clauses $C_1, \ldots , C_r$ in $n$ variables $x_1 , \ldots , x_n$, then any resolution refutation of $\textnormal{REF}^{F}_{s,t}$ has length greater than $2^{t^{\delta}}$.
\end{theorem}

The rest of this section is devoted to a proof of the theorem. We argue by contradiction. Fix $\epsilon>0$ and assume that for each $\delta>0$ and $t_0$ there are integers $n,r,s,t$ satisfying \eqref{eqn:conditions-on-nrst}, an unsatisfiable CNF $F$, and a resolution refutation $\Pi$ of $\textnormal{REF}^{F}_{s,t}$, such that $F$ consists of $r$ clauses $C_1, \ldots , C_r$ in $n$ variables $x_1, \ldots , x_n$, and $\Pi$ has length at most $2^{t^{\delta}}$.

The forthcoming distribution on partial assignments to the variables of $\textnormal{REF}^{F}_{s,t}$ employs in its definition and analysis two important parameters, $p$ and $w$. We choose them in function of $t$ and $\epsilon$ as follows:
\[ p = t^{-a} \textnormal{ with } a = \min \left\{\frac{2+\epsilon/2}{3+\epsilon/2}, \frac{3}{4} \right\}, \qquad w = t^{4/5}. \] 

We now fix values of $t_0, \delta$ for which we will get the desired contradiction. Take $t_0$ so large and $\delta>0$ so small that the inequalities
\begin{align} 
\label{eqn:requirement-on-t-and-delta} 
\max \left\{ e^{-\frac{pw}{3}} + 2s \cdot e^{- \frac{pt}{3}}, e^{-\frac{pt}{8r}} \right\} \cdot 2^{t^{\delta}} + 3s \cdot e^{-\frac{pt}{3}} + 3p + 67 p^3st & < 1, \\
\label{eqn:requirement-on-t} 
10pt +4w & < \frac{t}{4}, \\
\label{eqn:requirement-on-t-second}
e^{e^{\ln(t) -\frac{pt}{3}}}  & < 2,
\end{align}
hold for any $n,r,s,t$ satisfying \eqref{eqn:conditions-on-nrst}.

\begin{definition}
\label{def:home-pair-and-set-to}
For $i \in [s], j,j' \in [t], \ell \in [n], b \in \{0,1\}, m \in [r]$, we say that $(i,j)$ is the \emph{home pair} of the variable $D(i,j,\ell,b)$ (resp. $R(i,j,j')$; $L(i,j,j')$; $V(i,j,\ell')$; $I(j,m)$ if $i=1$).

We write $V(i,j,\cdot)$ to stand for the set $\{V(i,j,\ell) : \ell \in [n] \}$. Similarly, we write $I(j,\cdot), L(i,j,\cdot), R(i,j,\cdot)$ to stand for the corresponding sets of variables, and we denote by $D(i,j,\cdot,\cdot)$ the set of variables $\{D(i,j,\ell,b) : \ell \in [n], b \in \{0,1\}\}$. 

Let $\sigma$ be a partial assignment. We say that $V(i,j,\cdot)$ is \emph{set to} $\ell$ \emph{by} $\sigma$ if $\sigma(V(i,j,\ell)) = 1$ and $\sigma(V(i,j,\ell')) = 0$ for all $\ell' \in [n], \ell' \neq \ell$. Similarly for $I(j,\cdot), L(i,j,\cdot), R(i,j,\cdot)$. We say that $D(i,j,\cdot, \cdot)$ is \emph{set to} a clause $C_{i,j}$ \emph{by} $\sigma$ if for all $\ell \in [n], b \in \{0,1\}$, $\sigma(D(i,j,\ell,b)) = 1$ if $x_\ell^b \in C_{i,j}$ and $\sigma(D(i,j,\ell,b)) = 0$ if $x_\ell^b \not \in C_{i,j}$. 

For $Y \in \{D(i,j,\cdot,\cdot), V(i,j,\cdot), I(j,\cdot), R(i,j,\cdot), L(i,j,\cdot),\}$, we say that $Y$ is \emph{set by} $\sigma$ if $Y$ is set to $v$ by $\sigma$ for some value $v$. We will often omit saying ``by $\sigma$'' if $\sigma$ is clear from the context.
\end{definition}

\begin{definition} \label{def:random-restriction}
A \emph{random restriction} $\rho$ is a partial assignment to the variables of $\textnormal{REF}^{F}_{s,t}$ given by the following experiment:
\begin{enumerate}
\item For each pair $(i,j) \in [s] \times [t]$, with independent probability $p$ include $(i,j)$ in a set $A_D$. Then for each $(i,j) \in A_D$ and for each $\ell \in [n]$, independently, with probability 1/2 choose between including the literal $x_\ell$ or $\neg x_\ell$ in a clause $C_{i,j}$. Set $D(i,j,\cdot, \cdot)$ to $C_{i,j}$. 

\item For each $j \in [t]$, with independent probability $p$ include the pair $(1,j)$ in a set $A_I$. Then for each $(1,j) \in A_I \setminus A_D$, independently, choose at random $m \in [r]$ and set $I(j, \cdot)$ to $m$.

\item For each pair $(i,j) \in \{2, \ldots , s\} \times [t]$, with independent probability $p$ include $(i,j)$ in a set $A_V$. Then for each $(i,j) \in A_V$, independently, choose at random $\ell \in [n]$ and set $V(i,j, \cdot)$ to $\ell$.

\item For each pair $(i,j) \in \{2, \ldots , s\} \times [t]$, with independent probability $p$ include the pair $(i,j)$ in a set $A_{R\!L}$. Then, for each $i \in \{2, \ldots , s\}$, define $A_i := A_{R\!L} \cap (\{i\} \times [t])$ and do the following. If $|A_i| > 2pt$, define $h_i := \emptyset$, $B_{i-1} := \emptyset$. Otherwise, choose at random an injection $h_i$ from $\{L(i,j, \cdot): (i,j) \in A_i\} \cup \{R(i,j, \cdot): (i,j) \in A_i\}$ to $[t]$. Define $B_{i-1} := \{(i-1,j): j \in \im (h_i) \}$. Set $L(i,j, \cdot)$ to $h_i(L(i,j, \cdot))$ and $R(i,j, \cdot)$ to $h_i(R(i,j, \cdot))$ for all $(i,j) \in A_i$.
\end{enumerate}
\end{definition}

\begin{lemma} \label{lem:level-bounds-for-rho}
With probability at least $1 - 3s \cdot e^{-pt/3}$, all of the following are satisfied.
\begin{enumerate}[(i)]
\item \label{item:level-bounds-on-AD} For each $i \in [s]$, the cardinality of $A_D \cap (\{i\} \times [t])$ is at most $2pt$.
\item \label{item:level-bounds-on-ARL-and-AV} For each $i \in \{2, \ldots ,s\}$, the cardinality of $A_i$ is at most $2pt$ and the cardinality of $A_V \cap (\{i\} \times [t])$ is at most $2pt$.
\item \label{item:level-bounds-on-AI} The cardinality of $A_I$ is at most $2pt$.
\end{enumerate}
\end{lemma}
\begin{proof}
By the Chernoff bound and the union bound it follows that item \ref{item:level-bounds-on-AD} is false with probability at most $s \cdot e^{-pt/3}$. Similarly for the remaining items.
\end{proof}

\begin{definition}
Denote by $G_{\rho}$ the graph with vertices $A_D \cup A_V \cup A_I \cup A_{R\!L}\cup \bigcup_{i \in [s-1]}B_i$, and with edges only between vertices on neighboring levels, such that $(i,j)$ is connected by an edge to $(i-1,j')$ if and only if $h_i(L(i,j,\cdot)) = j'$ (then $(i-1,j')$ is called the \emph{left child} of $(i,j)$) or $h_i(R(i,j,\cdot)) = j'$ (then $(i-1,j')$ is the \emph{right child} of $(i,j)$).
\end{definition}

The following lemma will be used later to show that a random restriction likely does not falsify any clause of $\textnormal{REF}^{F}_{s,t}$.

\begin{lemma} \label{lem:forbid-in-rho-patterns}
With probability at least $1 - 3p - 67 p^3st$, the following are both satisfied.
\begin{enumerate}[(i)]
\item \label{item:forbid-in-rho-last-clause} $(s,t) \not \in (A_D \cup A_{R\!L} \cup A_V)$.
\item \label{item:forbid-in-rho-connected-3-tuple} There is no triple $\left( (i_1,j_1), (i_2,j_2),(i_3,j_3) \right)$ of elements of $[s] \times [t]$, such that all the following hold:
\begin{enumerate}[(a)]
\item \label{item:triple-in-G-rho} For each $u\! \in \! [3]$ there is $X \! \in \! \{ D,V,I,R\!L \}$ with $(i_u,j_u) 
\! \in \!  A_X$,
\item \label{item:triple-has-three-A-in-rho} $|\{ \left(i_u,j_u, X \right) :  u \in [3], X \! \in \! \{ D,V,I,R\!L \}, (i_u,j_u) \! \in \! A_X  \}| \geq 3$,
\item \label{item:triple-connected-in-rho} the subgraph of $G_{\rho}$ consisting of the vertices that are in the triple and their children and all edges that go from a vertex of the triple to its children, is connected.
\end{enumerate}
\end{enumerate}
\end{lemma}
\begin{proof}
The probability that item \ref{item:forbid-in-rho-last-clause} is true is $(1-p)^3 \geq 1-3p$.

Regarding item \ref{item:forbid-in-rho-connected-3-tuple}, we distinguish several cases based on the relative positions of the elements in a triple $((i_1,j_1), (i_2,j_2),(i_3,j_3))$. 
Note that the order in which the elements of the triple are listed does not matter in what we are proving, but some of the elements may coincide.
When considering the cases, recall that due to our choice of the function $h_i$ in the definition of $\rho$, two vertices in $G_{\rho}$ cannot share a child.

In case all the elements of the triple are the same, \ref{item:triple-has-three-A-in-rho} is satisfied only if the element is chosen to $A_X$ for three distinct values of $X$. This cannot happen on level 1, and on the other levels it happens with probability $p^3$. There are $st$ many triples considered in the present case, so by the union bound the probability that there is any such triple satisfying all conditions in \ref{item:forbid-in-rho-connected-3-tuple} is at most $p^3st$. 

In case $(i_1,j_1) \neq (i_2,j_2) = (i_3,j_3)$, condition \ref{item:triple-connected-in-rho} is satisfied only if $i_1 = i_2 +1$ or $i_2 = i_1 +1$. In each of these two subcases, there are at most $st^2$ such triples. 
In the former subcase, we must have $(i_1,j_1) \in A_{R\!L}$ and at the same time $h_{i_1}(R(i_1,j_1,\cdot)) = j_2$ or $h_{i_1}(L(i_1,j_1,\cdot)) = j_2$. This happens with probability at most $2p/t$. Also, $(i_2,j_2)$ has to be in $A_X$ and $A_{X'}$ for distinct $X,X'$, which happens with probability at most $3p^2$. So, the probability that any triple considered in this subcase satisfies \ref{item:triple-in-G-rho} - \ref{item:triple-connected-in-rho} is at most $st^2 \cdot 6p^3/t = 6p^3st$. 
In the latter subcase, $(i_2,j_2)$ has to be in $A_{R\!L}$, connected to $(i_1,j_1)$, and additionally it has to be in $A_D$ or $A_V$, while $(i_1,j_1)$ has to be in arbitrary possible $A_X$. This happens with probability at most $ 2p/t \cdot 2p \cdot 3p=12p^3/t$, so the probability that any such triple satisfies \ref{item:triple-in-G-rho} - \ref{item:triple-connected-in-rho} is at most $12p^3st$. 
 
 In case all the elements of the triple are distinct, we again consider two subcases: first, $i_1 = i_2+1 = i_3+2$, and second, $i_1 -1 = i_2 = i_3$. Each subcase concerns at most $st^3$ triples.
In the first subcase, $(i_3,j_3)$ has to be a child of $(i_2,j_2)$, which in turn has to be a child of $(i_1,j_1)$, and $(i_3,j_3)$ also has to be in arbitrary possible $A_X$. This happens with probability at most $12p^3/t^2$. Hence the probability that any such triple satisfies \ref{item:triple-in-G-rho} - \ref{item:triple-connected-in-rho} is at most $12p^3st$. 
In the second subcase, $(i_1,j_1)$ has to have children $(i_2,j_2)$ and $(i_3,j_3)$, and each child has to be in some $A_X$ for any suitable $X$. This happens with probability at most $2p/(t(t-1)) \cdot (3p)^2 = 18p^3/(t(t-1)) \leq 36p^3/t^2$. Hence the probability that any such triple satisfies \ref{item:triple-in-G-rho} - \ref{item:triple-connected-in-rho} is at most $36p^3st$. 
\end{proof}

We now define some specific ways to measure a clause and we use them in the next lemma to describe how a clause simplifies under a restriction.

\begin{definition}
\label{def:mentioned-and-important}
Let $E$ be a clause in $\Pi \! \restriction \! \rho$, and let $(i,j) \in [s] \times [t]$. If $E$ contains a literal of a variable from $D(i,j,\cdot, \cdot)$ (resp. $R(i,j,\cdot)$; $L(i,j,\cdot)$; $V(i,j,\cdot)$; $I(j, \cdot)$ and $i=1$), we say that the pair $(i,j)$ is $D$-\emph{mentioned} (resp. $R$-\emph{mentioned}; $L$-\emph{mentioned}; $V$-\emph{mentioned}; $I$-\emph{mentioned}) \emph{in} $E$.

We say that pair $(i,j)$ is \emph{mentioned in} $E$ if it is $Z$-mentioned in $E$ for some $Z \in \{D,V,I,R,L\}$. 

We say that $(i,j)$ is $V$-\emph{important} (resp. $L$-\emph{important}; $R$-\emph{important}; $I$-\emph{important}) \emph{in} $E$ if $E$ contains the negative literal of a variable in $V(i,j,\cdot)$ (resp. $L(i,j,\cdot); R(i,j,\cdot); I(j,\cdot )$ and $i=1$) or if $E$ contains at least $n/2$ (resp. $t/2$; $t/2$; $r/2$) positive literals of variables in $V(i,j,\cdot)$ (resp. $L(i,j,\cdot)$; $R(i,j,\cdot)$; $I(j,\cdot )$ and $i=1$). 
A pair is $D$-\emph{important} \emph{in} $E$ if it is $D$-mentioned in $E$. 
\end{definition}

\begin{lemma} \label{lem:widths-reduction}
With probability at least $1 - \max \left\{ e^{-\frac{pw}{3}} + 2s \cdot e^{- \frac{pt}{3}}, e^{-\frac{pt}{8r}} \right\} \cdot 2^{t^{\delta}}$, for every clause $E$ in $\Pi \! \restriction \! \rho$ all of the following are satisfied.
\begin{enumerate}[(i)]
\item \label{item:width-of-D-in-clause} At most $w$ many pairs $(i,j)$ are $D$-mentioned in $E$.
\item \label{item:width-of-I-in-clause} At most $w$ many pairs $(1,j)$ are $I$-important in $E$.
\item \label{item:width-of-V-in-clause} At most $w$ many pairs $(i,j)$ are $V$-important in $E$.
\item \label{item:width-of-L-in-clause} At most $w$ many pairs $(i,j)$ are $L$-important in $E$.
\item \label{item:width-of-R-in-clause} At most $w$ many pairs $(i,j)$ are $R$-important in $E$.
\item \label{item:width-of-m-image-of-I-in-clause} For each $m \in [r]$, $|\{ j: I(j,m) \in E \}| \leq \frac{t}{4}$.
\item \label{item:width-of-ell-image-of-V-in-clause} For each $i \in \{ s-n+1, \ldots , s-1 \}$ and $\ell \in [n]$, $|\{j: V(i,j,\ell) \in E\}| \leq \frac{t}{4}$.
\end{enumerate}
\end{lemma}

\begin{proof}
It is sufficient to prove that if $E'$ is a clause in $\Pi$ that violates any of \ref{item:width-of-D-in-clause} - \ref{item:width-of-ell-image-of-V-in-clause}, then with probability at least $1 - \max \left\{ e^{-\frac{pw}{3}} + 2s \cdot e^{- \frac{pt}{3}}, e^{-\frac{pt}{8r}} \right\}$, $E'$ is satisfied by $\rho$. Since $\Pi$ has length at most $2^{t^{\delta}}$, the lemma then follows by the union bound.

Regarding item \ref{item:width-of-D-in-clause}, assume that $E'$ in $\Pi$ $D$-mentions more than $w$ pairs $(i,j)$. This means that a literal of a variable in $D(i,j,\cdot, \cdot)$ is in $E'$ for more than $w$ many pairs $(i,j)$. For each such $(i,j)$, such a literal is satisfied by $\rho$ with probability at least $p/2$. So the probability that none of these literals in $E'$ is satisfied is at most $(1 - p/2)^w < e^{-pw/2}$.

Regarding item \ref{item:width-of-I-in-clause}, suppose that more than $w$ pairs $(1,j)$ are $I$-important in $E'$. For each such $(1,j)$, the probability that $(1,j) \in A_I \setminus A_D$ is $p(1-p)$, and provided this happens, the probability that $\rho$ satisfies a literal in $E'$ of a variable in $I(j,\cdot)$ is at least $\min \{(r-1)/r, 1/2\} = 1/2$. Hence the probability that $E'$ is not satisfied by $\rho$ is at most $(1-p(1-p)/2)^w < (1-p/3)^w < e^{-pw/3}$ (the first inequality follows from \eqref{eqn:requirement-on-t}).

Regarding item \ref{item:width-of-V-in-clause}, a calculation similar to that for \ref{item:width-of-I-in-clause} gives that a clause $E'$ in $\Pi$ with more than $w$ many $V$-important pairs $(i,j)$ is not satisfied by $\rho$ with probability at most $(1-p/2)^w < e^{-pw/2}$.

Regarding item \ref{item:width-of-L-in-clause}, suppose that more than $w$ many pairs $(i,j)$ from $\{2,\ldots,s\} \times [t]$ are $L$-important in $E'$. For each $i \in \{2,\ldots ,s\}$, assume without loss of generality that the set of pairs $(i,j)$ that are $L$-important in $E'$ is the set $\{(i,1), \ldots ,  (i,w_i) \}$; denote it by $W_i$. Note that the distribution of $\rho$ does not change if we choose $A_i$ and $h_i$ in $t$ many steps as follows. 
Start with $A_{i,0} = h_{i,0} = \emptyset$. At step $j = 1,2,\ldots,t$, first add $(i,j)$ to $A_{i,j-1}$ with probability $p$ to get $A_{i,j}$. Then, if $|A_{i,j}| \leq 2pt$ and $(i,j) \in A_{i,j}$, choose at random two distinct elements $j',j''$ from $[t] \setminus \im(h_{i,j-1})$, and define $h_{i,j} := h_{i,j-1} \cup \{ (L(i,j,\cdot),j'), (R(i,j,\cdot),j'') \}$. If $|A_{i,j}| \leq 2pt$ and $(i,j) \notin A_{i,j}$, define $h_{i,j} := h_{i,j-1}$. If $|A_{i,j}| > 2pt$ define $h_{i,j} := \emptyset$. This finishes step $j$. Finally, define $A_i := A_{i,t}$ and $h_i := h_{i,t}$. 

For $i \in \{2,\ldots ,s\}$, let $H_i$ be the set of literals in $E'$ of a variable in $L(i,j,\cdot)$ for some $(i,j) \in W_i$. Also, for $(i,j) \in W_i$, let $T_{i,j}$ be the set of those $j' \in [t]$ such that the partial assignment given by setting $L(i,j,\cdot)$ to $j'$ satisfies some literal in $H_i$. We know that $|T_{i,j}| \geq t/2$ for each $(i,j) \in W_i$.

The event that no literal in $H_i$ is satisfied by $\rho$ is a subset of the union of events (a) $|A_{i,t}| > 2pt$, and (b) $|A_{i,w_i}| \leq 2pt$ and for each $(i,j) \in A_{i,w_i}$, $h_{i,j}(L(i,j,\cdot)) \notin T_{i,j}$.
Event (a) happens with probability at most $e^{-pt/3}$ by the Chernoff bound. We bound the probability of event (b). For each $j \in [w_i]$, if $(i,j) \in A_{i,j}$ and $|A_{i,j}| \leq 2pt$, then the probability that $h_{i,j}(L(i,j,\cdot)) \in T_{i,j}$ is at least $(\left| T_{i,j} \setminus \im(h_{i,j-1}) \right| )/t \geq (t/2 - 4pt)/t = (1-8p)/2 \geq 1/3$ (the last inequality follows from \eqref{eqn:requirement-on-t}). Therefore, denoting $l := \min\{2pt, w_i\}$, the probability of event (b) is at most
\[
 \sum_{k=0}^{l} \binom{w_i}{k} p^k (1-p)^{w_i - k} \left( \frac{2}{3} \right)^k \leq \sum_{k=0}^{w_i}\binom{w_i}{k} \left( \frac{2p}{3} \right)^k (1-p)^{w_i - k} = (1- p/3)^{w_i}.
\]
Thus, the probability that no literal in $H_i$ is satisfied by $\rho$ is at most $e^{-pt/3} + e^{-pw_i / 3}$, and, denoting $S:= \{i \in \{2,\ldots , s\}: w_i \neq 0\}$, the probability that no literal in $\bigcup_{i \in S} H_i$ is satisfied by $\rho$ is at most 
\begin{align*}
 \prod_{i \in S} \left( e^{-\frac{pw_i}{3}} + e^{-\frac{pt}{3}} \right) & \leq e^{-\frac{pw}{3}} + \sum_{k = 1}^{\left| S \right|} \binom{\left| S \right|}{k} e^{- \frac{ptk}{3}} \\
 & \leq e^{-\frac{pw}{3}} + \left| S \right| \cdot e^{- \frac{pt}{3}} \sum_{k = 1}^{\left| S \right|} \binom{\left| S \right| -1}{k-1} e^{- \frac{pt(k-1)}{3}} \\
 & = e^{-\frac{pw}{3}} + \left| S \right| \cdot e^{- \frac{pt}{3}} \cdot \left(1+e^{- \frac{pt}{3}} \right)^{\left| S \right| - 1} \\
 & \leq e^{-\frac{pw}{3}} + s \cdot e^{- \frac{pt}{3}} \cdot e^{e^{\ln(t) -\frac{pt}{3}}} \leq e^{-\frac{pw}{3}} + 2s \cdot e^{- \frac{pt}{3}}, 
\end{align*}
where the penultimate inequality follows from $\left| S \right| - 1 \leq s \leq t$, and the last inequality follows from \eqref{eqn:requirement-on-t-second}.

Item \ref{item:width-of-R-in-clause} is handled in the same way as \ref{item:width-of-L-in-clause}.

Regarding item \ref{item:width-of-m-image-of-I-in-clause}, suppose that for some $m \in [r]$ there are more than $t/4$ of $I(j,m)$ in $E'$. Similarly to the case \ref{item:width-of-I-in-clause}, each such $I(j,m)$ is satisfied by $\rho$ with independent probability at least $p(1-p)/r > p/(2r)$, so $E'$ is not satisfied with probability at most $(1-p/(2r))^{t/4} < e^{-pt/(8r)}$.

Item \ref{item:width-of-ell-image-of-V-in-clause} is treated similarly to \ref{item:width-of-m-image-of-I-in-clause}, with the resulting probability of not satisfying $E'$ being $(1-p/n)^{t/4} < e^{-pt/(4n)} < e^{-pt/(8r)}$, where the last inequality follows from \eqref{eqn:conditions-on-nrst}.
\end{proof}
 
By \eqref{eqn:requirement-on-t-and-delta} and by Lemmas \ref{lem:level-bounds-for-rho}, \ref{lem:forbid-in-rho-patterns}, and \ref{lem:widths-reduction}, there is a restriction $\rho$ satisfying all the assertions of the lemmas. Fix any such $\rho$.

\begin{definition} \label{def:admissible}
A partial assignment $\sigma$ to the variables of $\textnormal{REF}^{F}_{s,t}$ is called an \emph{admissible assignment} if it extends $\rho$ and satisfies all the following conditions.
\begin{enumerate}[(C1)]
\item \label{item:admissible-set-or-untouched} For each $(i,j) \in [s] \times [t]$, $D(i,j,\cdot,\cdot)$ (resp. $V(i,j,\cdot)$, $I(j,\cdot)$, $L(i,j,\cdot)$, $R(i,j,\cdot)$) either is set to some clause (resp. some $\ell \in [n]$, some $m \in [r]$, some $j' \in [t]$, some $j' \in [t]$) by $\sigma$ or contains no variable that is in $\dom(\sigma)$. 

\item \label{item:admissible-RL-implies-D} For each $(i,j) \in [s] \times [t]$, if $L(i,j,\cdot)$ or $R(i,j,\cdot)$ is set to some  $j' \in [t]$, then both $D(i,j,\cdot,\cdot)$ and $D(i-1,j',\cdot,\cdot)$ are set.

\item \label{item:admissible-D-implies-VI} For each $(i,j) \in [s] \times [t]$, if $D(i,j,\cdot,\cdot)$ is set, then $V(i,j,\cdot)$ is set (if $i \in \{2, \ldots , s\}$) or $I(j,\cdot)$ is set (if $i = 1$).

\item \label{item:admissible-D-nontaut-and-fat} For each $(i,j) \in [s] \times [t]$, if $D(i,j,\cdot,\cdot)$ is set to a clause $C_{i,j}$, then $C_{i,j}$ is non-tautological and has at least $\min\{s-i, n\}$ many literals. If $D(i,j,\cdot,\cdot)$ is set to a clause $C_{i,j}$ with less than $n$ literals and $V(i,j,\cdot)$ is set to some $\ell \in [n]$, then none of the literals of $x_\ell$ is in $C_{i,j}$.

\item \label{item:admissible-empty-clause} If $D(s,t, \cdot , \cdot)$ is set, it is set to the empty clause.

\item \label{item:admissible-I}For each $j \in [t]$, if $D(1,j,\cdot,\cdot)$ and $I(j,\cdot)$ are set, then $\sigma$ satisfies all clauses in \eqref{refclause:axioms} with this $j$.

\item \label{item:admissible-RL-cut-variable} For each $i \in \{ 2, \ldots , s \}, j,j' \in [t]$, if $L(i,j,\cdot )$ (resp. $R(i,j,\cdot )$) is set to $j'$ and both  $V(i,j, \cdot )$, $D(i-1, j', \cdot , \cdot )$ are set, then $\sigma$ satisfies all clauses in \eqref{refclause:res-L-cut}  (resp. \eqref{refclause:res-R-cut}) with these $i,j,j'$.

\item \label{item:admissible-RL-transf} For each $i \in \{ 2, \ldots , s \}, j,j' \in [t]$, if $L(i,j,\cdot )$ (resp. $R(i,j,\cdot )$) is set to $j'$ and $V(i,j, \cdot )$, $D(i, j, \cdot , \cdot )$, $D(i-1, j', \cdot , \cdot )$ are set, then $\sigma$ satisfies all clauses in \eqref{refclause:res-L-transf}  (resp. \eqref{refclause:res-R-transf}) with these $i,j,j'$.

\item \label{item:admissible-RL-injection} For each $i \in \{ 2, \ldots , s \}$, the binary relation $h_{\sigma, i} :=\{ (Z(i,j, \cdot ),j') : j,j' \! \in \! [t], Z \! \in \! \{L,R\}, \text{and } Z(i,j,\cdot ) \text{ is set to } j' \text{ by } \sigma \}$ is a partial injection from $\{ Z(i,j, \cdot ) : j \in [t], Z \in \{L,R\} \}$ to $[t]$.
\end{enumerate}
\end{definition}

We now prove that an admissible assignment cannot falsify a clause of $\textnormal{REF}^{F}_{s,t} \! \restriction \! \rho$ (Lemma \ref{lem:admissible-not-falsifies-axiom}), that admissible assignments exist (Lemma \ref{lem:admissible-exists}), and that if for a clause $E$ in $\Pi \! \restriction \! \rho$ there is an admissible assignment that falsifies a literal in $E$ whenever it evaluates its variable and that evaluates each $D$- (resp. $V$-, $I$-, $L$-, $R$-) variable with a home pair $D$- (resp. $V$-, $I$-, $L$-, $R$-) important in $E$, then there is also an admissible assignment that does the same for at least one clause in $\Pi \! \restriction \! \rho$ from which $E$ was obtained by the resolution rule (Lemma \ref{lem:induction-step-admissible}). This is a contradiction, which concludes the proof of Theorem \ref{thm:main-lower-bound-REF-F-s-t}.

\begin{lemma} \label{lem:admissible-not-falsifies-axiom}
No clause in $\textnormal{REF}^{F}_{s,t} \! \restriction \! \rho$ is falsified by any admissible assignment.
\end{lemma}
\begin{proof}
Let $\sigma$ be an admissible assignment. It suffices to show that $\sigma$ does not falsify any clause of  $\textnormal{REF}^{F}_{s,t}$. This is guaranteed for each clause 
from \eqref{refclause:axioms} by \ref{item:admissible-set-or-untouched}, \ref{item:admissible-I}; 
from \eqref{refclause:nontaut} by \ref{item:admissible-set-or-untouched}, \ref{item:admissible-D-nontaut-and-fat}; 
from \eqref{refclause:res-L-cut} and \eqref{refclause:res-R-cut} by \ref{item:admissible-set-or-untouched}, \ref{item:admissible-RL-cut-variable};
from \eqref{refclause:res-L-transf} and \eqref{refclause:res-R-transf} by 
\ref{item:admissible-set-or-untouched}, \ref{item:admissible-RL-transf}; 
from \eqref{refclause:empty-clause} by \ref{item:admissible-set-or-untouched}, \ref{item:admissible-empty-clause}; 
and from \eqref{refclause:V-dom} - \eqref{refclause:R-func} by \ref{item:admissible-set-or-untouched}.
\end{proof}

\begin{lemma} \label{lem:admissible-exists}
There is an admissible assignment.
\end{lemma}

\begin{proof}
We first verify that $\rho$ satisfies the conditions of Definition \ref{def:admissible} except possibly for \ref{item:admissible-RL-implies-D}, \ref{item:admissible-D-implies-VI}. Then we extend $\rho$ by only assigning some $D,V,I$-variables to satisfy these remaining two conditions without violating the others.

Conditions \ref{item:admissible-set-or-untouched}, \ref{item:admissible-D-nontaut-and-fat}, \ref{item:admissible-I} and \ref{item:admissible-RL-injection} are satisfied by construction: for \ref{item:admissible-D-nontaut-and-fat} recall that for each $(i,j) \in [s] \times [t]$, if $D(i,j, \cdot , \cdot)$ is set to some clause by $\rho$ then the clause has exactly $n$ literals; for \ref{item:admissible-I} observe that its hypothesis is not satisfied by $\rho$; and for \ref{item:admissible-RL-injection} note that $h_{\rho,i} = h_i$ for $i \in \{2, \ldots , s\}$. Condition \ref{item:admissible-empty-clause} follows from item \ref{item:forbid-in-rho-last-clause} of Lemma \ref{lem:forbid-in-rho-patterns}. Item \ref{item:forbid-in-rho-connected-3-tuple} of the same lemma implies that neither the hypothesis in \ref{item:admissible-RL-cut-variable} nor the hypothesis in \ref{item:admissible-RL-transf} is met.

To extend $\rho$ to an admissible assignment, we distinguish cases based on the isomorphism type of the component in $G_{\rho}$ containing a pair for which \ref{item:admissible-RL-implies-D} or \ref{item:admissible-D-implies-VI} is not satisfied. By item \ref{item:forbid-in-rho-connected-3-tuple} of Lemma \ref{lem:forbid-in-rho-patterns}, there are only three types of components in $G_{\rho}$: 1) an isolated vertex, 2) a vertex with its two children and edges 
from the vertex to the children, and 3) two vertices with their children and edges from each of the two vertices to its children, such that one of the two vertices is a child of the other. 

In case 1), let $(i,j)$ be the isolated vertex. Only \ref{item:admissible-D-implies-VI} may be unsatisfied; assume this is the case. Recall again that whenever $\rho$ sets $D(i,j, \cdot ,\cdot)$ to some clause $C_{i,j}$, the clause contains $n$ literals and is non-tautological. Now, if $i \in \{2, \ldots , s\}$, set $V(i,j, \cdot)$ arbitrarily. For $i=1$, since $F$ is unsatisfiable, it must contain a clause $C_m$, for some $m \in [r]$, of which $C_{i,j}$ is a weakening. Set $I(j,\cdot)$ to $m$.

In case 2), since one of the three vertices of the component is in $A_{R\!L}$ and its children are not, item \ref{item:forbid-in-rho-connected-3-tuple} of Lemma \ref{lem:forbid-in-rho-patterns} implies that there is at most one triple $(i,j,Y)$ such that $(i,j)$ is a vertex of the three forming the component, $Y \in \{ D(i,j, \cdot ,\cdot), V(i,j, \cdot), I(j,\cdot) \}$, and $Y$ is set by $\rho$. We set the remaining $D(i,j, \cdot ,\cdot), V(i,j, \cdot), I(j,\cdot)$ for all vertices $(i,j)$ of the component in any way that respects how $Y$ is set, only uses clauses of $n$ literals to assign to the vertices (to satisfy \ref{item:admissible-D-nontaut-and-fat}), and satisfies \ref{item:admissible-I}, \ref{item:admissible-RL-cut-variable}, \ref{item:admissible-RL-transf} for these vertices.

In case 3), there are exactly two vertices of the component that are in $A_{R\!L}$, hence item \ref{item:forbid-in-rho-connected-3-tuple} of Lemma \ref{lem:forbid-in-rho-patterns} implies that there is no triple $(i,j,Y)$ such that $(i,j)$ is a vertex of the component, $Y \in \{ D(i,j, \cdot ,\cdot), V(i,j, \cdot), I(j,\cdot) \}$, and $Y$ is set by $\rho$. Hence we can set $D(i,j, \cdot ,\cdot), V(i,j, \cdot), I(j,\cdot)$ for all vertices $(i,j)$ of the component in any way that only uses clauses of $n$ literals to assign to the vertices (to satisfy \ref{item:admissible-D-nontaut-and-fat}), and satisfies \ref{item:admissible-I}, \ref{item:admissible-RL-cut-variable}, \ref{item:admissible-RL-transf} for these vertices.

Extending $\rho$ for every component of $G_{\rho}$ in this way satisfies \ref{item:admissible-set-or-untouched} - \ref{item:admissible-D-nontaut-and-fat}, does not affect \ref{item:admissible-empty-clause} and \ref{item:admissible-RL-injection}, and satisfies \ref{item:admissible-I}, \ref{item:admissible-RL-cut-variable}, \ref{item:admissible-RL-transf} because whenever we assigned all literals of a clause in \eqref{refclause:axioms}, \eqref{refclause:res-L-cut}, \eqref{refclause:res-R-cut}, \eqref{refclause:res-L-transf}, \eqref{refclause:res-R-transf}, we made sure the clause was satisfied.
\end{proof}

\begin{definition}
For a partial assignment $\sigma$ to the variables of $\textnormal{REF}^{F}_{s,t}$ that satisfies \ref{item:admissible-set-or-untouched} and \ref{item:admissible-RL-injection} of Definition \ref{def:admissible}, denote by $G_{\sigma}$ the graph with vertex set $\bigcup_{i \in [s-1]} \{(i,j) : j \in \im(h_{\sigma,i+1}) \} \cup \{(i,j) : (i,j) \textnormal{ is the home pair of a variable in } \dom(\sigma)\}$ and edges between $(i,j), (i',j') \in [s] \times [t]$ if and only if $i=i'+1$ and $h_{\sigma,i}(L(i,j,\cdot)) = j'$ or $h_{\sigma,i}(R(i,j,\cdot)) = j'$.
\end{definition}

\begin{lemma} \label{lem:induction-step-admissible}
Suppose that a clause $E$ in $\Pi \! \restriction \! \rho$ is obtained by the resolution rule from clauses $E_0$ and $E_1$. Suppose further that there is an admissible assignment $\sigma$ which satisfies both conditions
\begin{enumerate}[(i)]
\item \label{item:induction-step-admissible-rho-falsifies} every literal in $E$ of a variable in $\dom(\sigma)$ is falsified by $\sigma$,
\item \label{item:induction-step-admissible-rho-sets-important} for each $Z \in \{ D,V,I,R,L \}$, each $Z$-variable with a home pair $Z$-important in $E$ is in $\dom(\sigma)$.
\end{enumerate}
Then there is an admissible assignment $\tau$ and $b \in \{0,1\}$ such that \ref{item:induction-step-admissible-rho-falsifies} and \ref{item:induction-step-admissible-rho-sets-important} hold with $\tau$ in place of $\sigma$ and $E_b$ in place of $E$.
\end{lemma}

\begin{proof}
Let $\sigma$ be an admissible assignment that satisfies \ref{item:induction-step-admissible-rho-falsifies} and \ref{item:induction-step-admissible-rho-sets-important}. We first subject $\sigma$ to the following three-step cleanup process to obtain the minimal admissible sub-assignment $\sigma_1$ of $\sigma$ that satisfies \ref{item:induction-step-admissible-rho-falsifies} and \ref{item:induction-step-admissible-rho-sets-important}. 

Step 1: Remove from $\dom(\sigma)$ each variable in $L(i,j, \cdot)$ (resp. $R(i,j,\cdot))$ that is not in $\dom(\rho)$ such that $(i,j)$ is not $L$-important (resp. $R$-important) in $E$. Denote by $\sigma'$ the new partial assignment.

Step 2: Remove from $\dom(\sigma')$ each variable in $D(i,j,\cdot, \cdot)$ that is not in $\dom(\rho)$ such that $(i,j)$ is not $D$-important in $E$ and no edge in $G_{\sigma'}$ is incident to $(i,j)$. Denote by $\sigma''$ the new partial assignment.

Step 3: For each $(i,j) \in [s] \times [t]$, remove from $\dom(\sigma'')$ each variable in $V(i,j,\cdot, \cdot)$ (resp. $I(j, \cdot)$ if $i=1$) that is not in $\dom(\rho)$ such that $(i,j)$ is not $V$-important (resp. $I$-important) in $E$ and $D(i,j, \cdot , \cdot)$ is not set by $\sigma''$. 
Let $\sigma_1$ stand for the resulting partial assignment.

It is straightforward to check that that $\sigma', \sigma'', \sigma_1$ are admissible assignments (the order of the steps was chosen to maintain  \ref{item:admissible-RL-implies-D} and \ref{item:admissible-D-implies-VI} satisfied during the process; \ref{item:admissible-set-or-untouched} follows since we always unassign variables in groups listed there, and the remaining conditions of Definition \ref{def:admissible} cannot turn from being true to false by unassigning variables), and that they satisfy \ref{item:induction-step-admissible-rho-falsifies} and \ref{item:induction-step-admissible-rho-sets-important}. The three steps and the order of their execution ensure that $\sigma_1$ is the minimal admissible sub-assignment of $\sigma$ satisfying \ref{item:induction-step-admissible-rho-falsifies} and \ref{item:induction-step-admissible-rho-sets-important}. 

Let $Q$ be the variable resolved on to obtain $E$ from $E_0$ and $E_1$. Suppose that $Q$ is a $Z$-variable, $Z \in \{ D,V,I,L,R \}$, with a home pair $(i,j) \in [s] \times [t]$.

If $Q \in \dom(\sigma_1)$, then $\sigma_1$ with either $E_0$ or $E_1$ already satisfy \ref{item:induction-step-admissible-rho-falsifies} and \ref{item:induction-step-admissible-rho-sets-important}. 

If $Q \not \in \dom(\sigma_1)$ and $(i,j)$ is not $Z$-important in $E \cup \{Q\}$, then $(i,j)$ cannot be $Z$-important in $E_b$ either, where $b \in \{0,1\}$ is such that $ Q \in E_b$, and therefore $\sigma_1$ with $E_b$ already satisfy \ref{item:induction-step-admissible-rho-falsifies} and \ref{item:induction-step-admissible-rho-sets-important}. 

Otherwise, we have that $Q \not \in \dom(\sigma_1)$ and $(i,j)$ is $Z$-important in $E \cup \{Q\}$. It is enough to show how to extend $\sigma_1$ to an admissible assignment $\tau$ which assigns a value to $Q$ such that \ref{item:induction-step-admissible-rho-falsifies} and \ref{item:induction-step-admissible-rho-sets-important} are satisfied with $\tau$ in place of $\sigma$ and $E \cup \{Q^{1-\tau(Q)} \}$ in place of $E$. 

Observe that $(i,j)$ is not $Z$-important in $E$. This is because $\sigma_1$ with $E$ satisfy \ref{item:induction-step-admissible-rho-sets-important} and $Q \not \in \dom(\sigma_1)$. We consider three cases.

Case 1. Suppose that $Q \in V(i,j,\cdot)$ (resp. $Q \in I(j, \cdot)$ and $i=1$). 
Because $(i,j)$ is not $V$- (resp. $I$-) important in $E$, there are less than $n/2$ (resp. $r/2$) positive and no negative literals of variables from $V(i,j,\cdot)$ (resp. $I(j,\cdot)$) in $E$. Pick any $\ell \in [n]$ such that $V(i,j,\ell) \not \in E \cup \{Q\}$ (resp. any $m \in [r]$ such that $I(j,m) \not \in E \cup \{Q\}$) and extend $\sigma_1$ to $\tau$ by setting $V(i,j,\cdot)$ to $\ell$ (resp. $I(j, \cdot)$ to $m$). This choice makes $\tau$ with $E \cup \{Q\}$ satisfy \ref{item:induction-step-admissible-rho-falsifies} and \ref{item:induction-step-admissible-rho-sets-important}. From the construction and the fact that $\sigma_1$ is an admissible assignment it follows that $\tau$ is too. In particular, to see that \ref{item:admissible-D-nontaut-and-fat} - \ref{item:admissible-RL-transf} are satisfied by $\tau$, note that since $\sigma_1$ satisfies \ref{item:admissible-RL-implies-D}, \ref{item:admissible-D-implies-VI} and $Q \not \in \dom(\sigma_1)$, no variable from $D(i,j, \cdot ,\cdot)$ is in $\dom(\sigma_1)$ (and hence is not in $\dom(\tau)$ either), and there is no edge in $G_{\sigma_1}$ incident to $(i,j)$ (and hence there is no such edge in $G_{\tau}$ either).

Case 2. Suppose that $Q \in D(i,j,\cdot, \cdot)$. Since $(i,j)$ is not $D$-important in $E$, no literal of a variable from $D(i,j,\cdot, \cdot)$ is in $E$. 
Because $\sigma_1$ satisfies \ref{item:admissible-RL-implies-D} and $Q \not \in \dom(\sigma_1)$, there is no edge of $G_{\sigma_1}$ incident to $(i,j)$. But $V(i,j,\cdot)$ (resp. $I(j,\cdot)$ if $i=1$) may be set by $\sigma_1$. 
If $(i,j) \in \{2,\ldots, s\} \times [t]$, set $D(i,j,\cdot, \cdot)$ to an arbitrary non-tautological clause with $n$ literals, unless $(i,j) = (s,t)$, in which case set $D(i,j,\cdot, \cdot)$ to the empty clause. Then, set $V(i,j,\cdot)$, unless it is already set by $\sigma_1$, to any value $\ell \in [n]$ such that $V(i,j,\ell) \not \in E$. Such $\ell$ exists, because if $V(i,j,\cdot)$ is not set by $\sigma_1$ then $(i,j)$ is not $V$-important in $E$, and hence there are more than $n/2$ available values to choose $\ell$ from. 
If $i=1$, either $I(j, \cdot)$ is set by $\sigma_1$ to some $m \in [r]$ and we set $D(1,j,\cdot, \cdot)$ to any non-tautological clause with $n$ literals that contains the literals of $C_m$, or $I(j, \cdot)$ is not set by $\sigma_1$, in which case we first set it to any $m \in [r]$ such that $I(j,m) \not \in E$ and then we set $D(1,j,\cdot, \cdot)$ as before; such $m$ exists because if $I(j,\cdot)$ is not set by $\sigma_1$ then $(1,j)$ is not $I$-important in $E$, and hence there are more than $r/2$ available values to choose $m$ from. Like in the previous case, it is easy to check that in each of the subcases considered we extended $\sigma_1$ to an admissible assignment $\tau$ such that $\tau$ with $E \cup \{Q^{1-\tau(Q)} \}$ satisfy 
\ref{item:induction-step-admissible-rho-falsifies} and \ref{item:induction-step-admissible-rho-sets-important}.

Case 3. Suppose that $Q \in L(i,j,\cdot)$ (if $Q \in R(i,j,\cdot)$ we proceed in a completely analogous way). 
We may assume that $V(i,j,\cdot)$ is set to some clause $C_{i,j}$ and $D(i,j,\cdot, \cdot)$ is set to some $\ell \in [n]$ by $\sigma_1$; if not, perform the steps in Case 2 to set them both.
We have to set $L(i,j,\cdot)$ to some $j'$, i.e., we have to add to $G_{\sigma_1}$ an edge from $(i,j)$ to a left child $(i-1,j')$, and the first set $U_1$ of pairs $(i-1,j')$ we would like to avoid are the vertices of $G_{\sigma_1}$. To upper bound the number of vertices of its subgraph $G_{\rho}$ that are on level $i-1$, we use items \ref{item:level-bounds-on-AD} - \ref{item:level-bounds-on-AI} of Lemma \ref{lem:level-bounds-for-rho}. According to these items, $G_{\rho}$ has on level $i-1$: at most $|A_{i-1}| \leq 2pt$ endpoints of edges between levels $i-1$ and $i-2$, further, at most $|B_{i-1}| = 2|A_i| \leq 4pt$ endpoints of edges between levels $i$ and $i-1$, and at most $|A_D|+ |A_V| \leq 4pt$ (or $|A_D|+|A_I|\leq 4pt$ if $i-1=1$) isolated vertices. 
To upper bound the number of vertices in $G_{\sigma_1}$ on level $i-1$ that are not in $G_{\rho}$, note that by the minimality of $\sigma_1$, each such vertex either is or shares an edge with a $Z'$-important pair in $E$ for some $Z' \in \{D,V,I,R,L\}$. By items \ref{item:width-of-D-in-clause} - \ref{item:width-of-R-in-clause} of Lemma \ref{lem:widths-reduction}, there are at most $4w$ pairs that can in this way give rise to a vertex in $G_{\sigma_1}$ on level $i-1$ that is not in $G_{\rho}$. Hence $|U_1| \leq 10pt + 4w$.

The second set $U_2$ of pairs $(i-1,j')$ we would like to avoid when looking for a suitable left child of $(i,j)$ are those with $L(i,j,j') \in E$ (because we want $\tau$ to satisfy \ref{item:induction-step-admissible-rho-falsifies}). Because $(i,j)$ is not $L$-important in $E$, we have $|U_2| < t/2$.

The third set $U_3$ of pairs $(i-1,j')$ we would like the left child of $(i,j)$ to avoid depends on whether $i=2$ or $i \in \{3, \ldots ,s\}$. If $i=2$, since $\sigma_1$ satisfies \ref{item:admissible-D-nontaut-and-fat}, we know that $C_{2,j}$ has at least $n-1$ literals (because $s \geq n+1$) and that the clause at the left child of $(2,j)$ is completely determined by $C_{2,j}$ (because we want $\tau$ to satisfy \ref{item:admissible-D-nontaut-and-fat}): it is the clause $(C_{2,j} \setminus \{ \neg x_\ell\}) \cup \{x_\ell\}$. 
Pick some $m \in [r]$ such that the clause $C_m$ of $F$ is a subset of $(C_{2,j} \setminus \{ \neg x_\ell\}) \cup \{x_\ell\}$.
For some $j' \in [t]$, we want to set $D(i-1,j', \cdot, \cdot)$ to $(C_{2,j} \setminus \{ \neg x_\ell\}) \cup \{x_\ell\}$ and $I(j',\cdot)$ to $m$ by $\tau$ (in order to satisfy \ref{item:admissible-I}), but this is not possible if $I(j',m) \in E$. For this reason, in the case $i=2$ we define $U_3 := \{(1,j') : I(j',m) \in E \}$. By item \ref{item:width-of-m-image-of-I-in-clause} of Lemma \ref{lem:widths-reduction}, $|U_3| \leq t/4$. 

If $i \in \{3, \ldots ,s\}$, we are concerned with the case that $C_{i,j}$ has less than $n-1$ literals; otherwise we leave $U_3$ empty. Since $\sigma_1$ satisfies \ref{item:admissible-D-nontaut-and-fat}, no literal of $x_\ell$ is in $C_{i,j}$, and $C_{i,j}$ has at least $s-i$ literals. Pick some $\ell' \in [n]$ such that no literal of $x_{\ell'}$ is in $C_{i,j} \cup \{x_\ell \}$. For some $j' \in [t]$, we want to set $D(i-1,j', \cdot, \cdot)$ to $C_{i,j} \cup \{x_\ell \}$ and $V(i-1,j')$ to $\ell'$ (to make $\tau$ satisfy \ref{item:admissible-D-nontaut-and-fat}). But this is not possible to do if $V(i-1,j',\ell') \in E$. Therefore, in the case $i \in \{3, \ldots ,s\}$ we define $U_3 := \{(i-1,j') : V(i-1,j',\ell') \in E \}$. Thanks to item \ref{item:width-of-ell-image-of-V-in-clause} of Lemma \ref{lem:widths-reduction}, we have $|U_3| \leq t/4$ again.

Now set $L(i,j,\cdot)$ to $j'$ such that $(i-1,j') \not \in U_1 \cup U_2 \cup U_3$. Such $j'$ exists because $|U_1 \cup U_2 \cup U_3| \leq 10pt+4w + t/2 + t/4 < t$ by \eqref{eqn:requirement-on-t}. Also, set $D(i-1,j',\cdot, \cdot)$ and either $I(j',\cdot)$ (if $i = 2$) or $V(i-1,j',\cdot)$ (if $i \in \{3,\ldots, s\}$) as indicated at the definition of $U_3$. In the case where we left $U_3$ empty, set $D(i-1,j', \cdot, \cdot)$ to $(C_{i,j} \setminus \{ \neg x_\ell \}) \cup \{ x_\ell \}$ and set $V(i-1,j',\cdot)$ to any $\ell' \in [n]$ such that $V(i-1,j',\ell') \not \in E$; such $\ell'$ exists because $(i-1,j')$ is not $V$-important in $E$ (by avoiding $U_1$). This finishes the definition of $\tau$. 
Item \ref{item:induction-step-admissible-rho-sets-important} is satisfied by $\tau$ and $E \cup \{ Q^{1-\tau(Q)} \}$ because $D(i,j, \cdot, \cdot)$ is set by $\tau$. 
Item \ref{item:induction-step-admissible-rho-falsifies}
follows for the variables in $D(i,j,\cdot, \cdot)$, $V(i,j,\cdot)$ because we set them using Case 2; for the variables in $L(i,j,\cdot)$ because $(i-1,j') \not \in U_2$; and for the variables in $D(i-1,j',\cdot, \cdot)$, $V(i-1,j',\cdot)$, $I(j',\cdot)$ by avoiding $U_3$ and because $(i-1,j')$ is neither $D$- nor $V$- nor $I$-important in $E$ (due to avoiding $U_1$). Finally, $\tau$ is an admissible assignment: the reasons why \ref{item:admissible-D-nontaut-and-fat} and \ref{item:admissible-I} are satisfied at $(i-1,j')$ are given at the definition of $U_3$, and the remaining conditions are easy to check due to avoiding $U_1$.
\end{proof}

\begin{remark}
\label{remark:bin-encoding}
If we assume $s=n+1$ in Theorem \ref{thm:main-lower-bound-REF-F-s-t} (instead of assuming only $s \geq n+1$) then we can allow $t$ to be smaller: it is enough to assume that $t \geq r^{2+\epsilon}$. This can be useful if one wants to reduce the number of variables of $\textnormal{REF}^{F}_{s,t}$ while keeping the lower bound of the theorem valid. The latter can be showed by making only the following modification in the proof of Theorem \ref{thm:main-lower-bound-REF-F-s-t}: change the definition of $p$ to $p = s^{-1/3}t^{-a'}$ with $a' = \min \left\{ \frac{1+ \epsilon}{3 + \epsilon}, \frac{1}{2} \right\}$, and change the definition of $w$ to $w = s^{1/3}t^{3/5}$. 

We note that if in the definition of $\textnormal{REF}^{F}_{s,t}$ we encode the functions determined by $V$- and $I$-variables in binary instead of in unary, the assumption $t \geq r^{3+\epsilon}$ in Theorem \ref{thm:main-lower-bound-REF-F-s-t} is not necessary (and the proof of the theorem simplifies somewhat), and, in addition, the $L$- and $R$-variables can be encoded in binary too (with some further simplifications of the proof). This reduces the number of variables of $\textnormal{REF}^{F}_{s,t}$ in two ways, by allowing a smaller $t$ and by using a more efficient encoding.
\end{remark}

\begin{remark}
\label{remark:obstacles}
Most of the obstacles our proof has to overcome are caused by the nature of the object described by $\textnormal{REF}^{F}_{s,t}$ and by the fact that the functions determined by $V,I,L,R$-variables are encoded in unary, rather than in binary. This forces us to work with several notions of width of two kinds, and we cannot keep as an invariant of the maintained partial assignment that it falsifies all literals of a clause as we traverse the refutation (as is the case e.g. in \cite{Thapen2016}). Moreover, keeping falsified just the literals with important indices and adding some simple conditions about not directly falsifying an axiom (a method which works e.g. in \cite{Pudlak2000} for the pigeonhole principle) is not enough either, because we need to be prepared to consistently answer the prover's questions about clauses situated at remote parts of the same not too small component (learnt through the $L$- and $R$-variables). This is further complicated by the need to respond by adding a fresh literal to a clause that has too few literals to make sure its width grows fast enough (such clauses originate in the component of the empty clause), and by the necessity to arrive to a weakening of a clause in $F$ when asked how a clause on level 2 is derived; both are more difficult to meet under the unary encoding and pose specific requirements on random restrictions. Our strategy stores some useful information in the form of negating some other literals than just those with important indices in a clause, as can be seen in the hierarchy of setting of variables of different kinds in Definition \ref{def:admissible}.
\end{remark}

\section{Reflection Principle for Resolution}
\label{sec:refl-princ-for-res}
We express the negation of the reflection principle for Resolution by a CNF in the form of a conjunction $\textnormal{SAT}^{n,r} \land \textnormal{REF}^{n,r}_{s,t}$. The only shared variables by the formulas $\textnormal{SAT}^{n,r}$ and $\textnormal{REF}^{n,r}_{s,t}$ encode a CNF with $r$ clauses in $n$ variables. The meaning of $\textnormal{SAT}^{n,r}$ is that the encoded CNF is satisfiable, while the meaning of $\textnormal{REF}^{n,r}_{s,t}$ is that it has a resolution refutation of $s$ levels of $t$ clauses. A formal definition is given next.

Formula $\textnormal{SAT}^{n,r}$ has the following variables. $C$-\emph{variables} $C(m,\ell,b)$, $m \in [r], \ell \in [n], b \in \{0,1\},$ encode clauses $C_m$ as follows: $C(m,\ell,1)$ (resp. $C(m,\ell,0)$) means that the literal $x_\ell$ (resp. $\neg x_\ell$) is in $C_m$. 
$T$-\emph{variables} $T(\ell)$, $\ell \in [n]$, and $T(m,\ell,b)$, $m \in [r], \ell \in [n], b \in \{0,1\}$, encode that an assignment to variables $x_1,\ldots, x_n$ satisfies the CNF $\{C_1,\ldots, C_r\}$. The meaning of $T(\ell)$ is that the literal $x_\ell$ is satisfied by the assignment. The meaning of $T(m,\ell,1)$ (resp. $T(m,\ell,0)$) is that clause $C_m$ is satisfied through the literal $x_\ell$ (resp. $\neg x_\ell$).

We list the clauses of $\textnormal{SAT}^{n,r}$:
\begin{align}
\label{satclause:at-least-one-lit-sat} & T(m,1,1) \lor T(m,1,0) \lor \ldots \lor T(m,n,1) \lor T(m,n,0)& m \in [r],\\
\label{satclause:sat-by-positive} & \neg T(m,\ell,1) \lor T(\ell) & m \in [r], \ell \in [n],\\
\label{satclause:sat-by-negative} & \neg T(m,\ell,0) \lor \neg T(\ell) & m \in [r], \ell \in [n],\\
\label{satclause:sat-lit-in-clause} & \neg T(m,\ell,b) \lor C(m,\ell,b) & m \in [r], \ell \in [n], b \in \{0,1\}.
\end{align}
The meaning of \eqref{satclause:at-least-one-lit-sat} is that clause $C_m$ is satisfied through at least one literal.
The meaning of \eqref{satclause:sat-by-positive} and \eqref{satclause:sat-by-negative} is that if $C_m$ is satisfied through a literal, then the literal is satisfied.
The meaning of \eqref{satclause:sat-lit-in-clause} is that if $C_m$ is satisfied through a literal, then it contains the literal. 

Variables of $\textnormal{REF}^{n,r}_{s,t}$ are the variables $C(m,\ell,b)$ of $\textnormal{SAT}^{n,r}$ together with all the variables of $\textnormal{REF}^F_{s,t}$ for some (and every) $F$ of $r$ clauses in $n$ variables. That is, $\textnormal{REF}^{n,r}_{s,t}$ has the following variables:
$C(m,\ell,b)$ for $m \in [r], \ell \in [n], b \in \{0,1\}$;
$D(i,j,\ell,b)$ for $i \in [s], j \in [t], \ell \in [n], b \in \{0,1\}$; 
$R(i,j,j')$ and $L(i,j,j')$ for $i \! \in \! [s] \! \setminus \! \{1\}, j,j' \in [t]$; $V(i,j,\ell)$ for $i \! \in \! [s] \! \setminus \! \{1\}, j \in [t], \ell \in [n]$; $I(j,m)$ for $j \in [t], m \in [r]$.

The clauses of $\textnormal{REF}^{n,r}_{s,t}$ are 
\eqref{refclause:nontaut} - \eqref{refclause:R-func} of  $\textnormal{REF}^F_{s,t}$ together with the following clauses (to replace clauses \eqref{refclause:axioms}): 
\begin{align}
\label{refclause-n-r:axioms}
& \neg I(j,m) \lor \neg C(m,\ell,b) \lor D(1,j,\ell,b) &  j \in [t], m \in [r], \ell \in [n], b \in \{0,1\},
\end{align}
saying that if clause $C_{1,j}$ is a weakening of clause $C_m$, then the former contains each literal of the latter. So the difference from \eqref{refclause:axioms} is that $C_m$ is no longer a clause of some fixed formula $F$, but it is described by $C$-variables.

\begin{lemma}
\label{lem:subst-to-SAT}
Let $F$ be a CNF with $r$ clauses $C_1,\ldots, C_r$ in $n$ variables $x_1,\ldots, x_n$, and let $\gamma_F$ be an assignment such that its domain are all $C$-variables and $\gamma_F (C(m,\ell,b)) = 1$ if $x_\ell^b \in C_m$ and $\gamma_F (C(m,\ell,b)) = 0$ if $x_\ell^b \notin C_m$. There is a substitution $\tau$ that maps the variables of $\textnormal{SAT} \! \restriction \! \gamma_F$ to $\{0,1\} \cup \{ x_\ell^b : \ell \in [n], b \in \{0,1\} \}$ such that $( \textnormal{SAT} \! \restriction \! \gamma_F ) \! \restriction \! \tau$ is $F$ together with some tautological clauses. 
\end{lemma}

\begin{proof}
Define $\tau$ as follows. If $\gamma_F(C(m,\ell,b)) = 0$, then $\tau(T(m,\ell,b)) = 0$. This satisfies \eqref{satclause:sat-by-positive} 
 - \eqref{satclause:sat-lit-in-clause} and deletes $T(m,\ell,b)$ from \eqref{satclause:at-least-one-lit-sat}. If $\gamma_F(C(m,\ell,b)) = 1$, then \eqref{satclause:sat-lit-in-clause} has been satisfied and we define $\tau(T(m,\ell,b)) = x_\ell^b$ and $\tau(T(\ell)) = x_\ell$. This choice turns \eqref{satclause:sat-by-positive} - \eqref{satclause:sat-by-negative} into tautological clauses and correctly substitutes the remaining literals of \eqref{satclause:at-least-one-lit-sat} to yield the clause $C_m$ of $F$.
\end{proof}

A polynomial size $\text{Res}(2)$ upper bound for $\textnormal{SAT}^{n,r} \land \textnormal{REF}^{n,r}_{s,t}$, is proved in Section \ref{sec:refl-princ-upper-bound}. We now prove the lower bound, stated in the introduction as Theorem \ref{thm:refl-princ-Lower-bound-Introduction} and restated below as Theorem \ref{thm:refl-princ-lower-bound}.

\begin{theorem}
\label{thm:refl-princ-lower-bound}
For every $c>4$ there is $\delta >0$ and an integer $n_0$ such that if $n,r,s,t$ are integers satisfying 
\begin{equation}
\label{eq:refl-princ-lower-bound-parameters}
 t\geq s \geq n+1,\qquad r \geq n \geq n_0,\qquad  n^c \geq t \geq r^4,
\end{equation}
then any resolution refutation of $\textnormal{SAT}^{n,r} \land \textnormal{REF}^{n,r}_{s,t}$ has length greater than $2^{n^{\delta}}$.
\end{theorem}
\begin{proof}
Fix $c>4$. We first observe that if $\Pi$ is a resolution refutation of $\textnormal{SAT}^{n,r} \land \textnormal{REF}^{n,r}_{s,t}$ and $\gamma$ is a partial assignment such that its domain are all $C$-variables, then $\Pi \! \restriction \! \gamma$ is either a refutation of $\textnormal{REF}^{n,r}_{s,t} \! \restriction \! \gamma$, or a refutation of $\textnormal{SAT}^{n,r} \! \restriction \! \gamma$. This is because $\Pi \! \restriction \! \gamma$ is a resolution refutation and the two restricted formulas do not share any variables. 

Let $F$ be a CNF with $r$ clauses in $n$ variables, and let $\gamma_F$ be a partial assignment defined in Lemma \ref{lem:subst-to-SAT}, which evaluates the $C$-variables so that they describe the clauses of $F$. Notice that $\textnormal{REF}^{n,r}_{s,t} \! \restriction \! \gamma_F$ is $\textnormal{REF}^{F}_{s,t}$, since $\gamma_F$ turns the clauses \eqref{refclause-n-r:axioms} into the clauses \eqref{refclause:axioms} (and removes the satisfied clauses). Therefore, in the case that $\Pi \! \restriction \! \gamma_F$ is a refutation of $\textnormal{REF}^{n,r}_{s,t} \! \restriction \! \gamma_F$ and $F$ is unsatisfiable, the lower bound of Theorem \ref{thm:main-lower-bound-REF-F-s-t} applies (setting $\epsilon = 1$ in that theorem, there is $n_0$ such that conditions \eqref{eqn:conditions-on-nrst} on $n,r,s,t$ follow from \eqref{eq:refl-princ-lower-bound-parameters}): the theorem yields some $\delta_1>0$ such that the length of $\Pi \! \restriction \! \gamma_F$ is at least $2^{n^{\delta_1}}$.

On the other hand, if $\Pi \! \restriction \! \gamma_F$ is a refutation of $\textnormal{SAT}^{n,r} \! \restriction \! \gamma_F$, the substitution $\tau$ from Lemma \ref{lem:subst-to-SAT} takes it into a not larger resolution refutation of $F$ (since tautological clauses can be removed from any resolution refutation).
 
It remains to take any unsatisfiable formula $F$ whose number of clauses is polynomially related to the number of variables and that requires resolution refutations of exponential length, e.g. negation of the pigeonhole principle \cite{Haken1985}. A trivial modification of $F$ to serve also in the extreme case $r=n$ allowed by \eqref{eq:refl-princ-lower-bound-parameters} yields $\delta_2>0$ such that any resolution refutation of $F$ has length greater than $2^{n^{\delta_2}}$, where $n$ is the number of variables of $F$. 

Setting $\delta$ to the minimum of $\delta_1$ and $\delta_2$ concludes the proof of the theorem.
\end{proof}

A similar proof gives Theorem \ref{thm:examples-separating-Introduction}. We restate the theorem below for convenience.
\begin{theorem}
\label{thm:examples-separating}
Let $\delta_1 > 0$ and let $\{A_n\}_{n \geq 1}$ be a family of unsatisfiable CNFs such that $A_n$ is in $n$ variables, has the number of clauses polynomial in $n$, and has no resolution refutations of length at most $2^{n^{\delta_1}}$. Then there is $\delta > 0$ and a polynomial $p$ such that $A_n \land \textnormal{REF}^{A_n}_{n+1,p(n)}$ has no resolution refutations of length at most $2^{n^{\delta}}$ and has polynomial size $\textnormal{Res}(2)$ refutations.
\end{theorem}

\begin{proof}
Let $p(n) \geq \max\{r^4,t_0\}$, where $r$ is the maximum of the number of clauses of $A_n$ and $n$, and $t_0$ is given by Theorem \ref{thm:main-lower-bound-REF-F-s-t} for $\epsilon = 1$. That theorem and the assumptions on $A_n$ give the required lower bound. To get the upper bound, start with the $\textnormal{Res}(2)$ refutation of $\textnormal{SAT}^{n,r} \land \textnormal{REF}^{n,r}_{n+1,p(n)}$ given by Theorem \ref{thm:refl-princ-Upper-bound-Introduction}. Take the substitutions $\gamma_{A_n}$ and $\tau$ from Lemma \ref{lem:subst-to-SAT} and observe that $((\textnormal{SAT}^{n,r} \land \textnormal{REF}^{n,r}_{n+1,p(n)}) \! \restriction \! \gamma_{A_n}) \! \restriction \! \tau$ is $A_n \land \textnormal{REF}^{A_n}_{n+1,p(n)}$ together with some tautological clauses.
\end{proof}

\section{The Upper Bounds}
\label{sec:refl-princ-upper-bound}

We restate and prove Theorem \ref{thm:refl-princ-Upper-bound-Introduction} from the Introduction.

\begin{theorem}
\label{thm:refl-princ-Upper-bound}
The negation of the reflection principle for Resolution expressed by the formula $\textnormal{SAT}^{n,r} \land \textnormal{REF}^{n,r}_{s,t}$ has $\text{Res}(2)$ refutations of size $O(trn^2+ tr^2 + st^2n^3 + st^3n)$.
\end{theorem}

\begin{proof}
By induction on $i \in [s]$ we derive for each $j \in [t]$ the formula
\begin{equation}
 D_{i,j} := \bigvee_{\ell \in [n], b \in \{0,1\}} \left( D(i,j,\ell,b) \land T(\ell)^b \right).
\end{equation}
Then, cutting $D_{s,t}$ with \eqref{refclause:empty-clause} for each $\ell \in [n]$ and $b \in \{0,1\}$, yields the empty clause.

Base case: For each $j \in [t]$ we shall derive $D_{1,j}$. For each $m \in [r], \ell \in [n], b \in \{0,1\}$, cut \eqref{satclause:sat-lit-in-clause} with \eqref{refclause-n-r:axioms} to obtain $\neg I(j,m) \lor \neg T(m,\ell,b) \lor D(1,j,\ell,b)$. Applying $\land$-introduction to this and $\neg T(m,\ell,b) \lor T(\ell)^b$ (which is either
\eqref{satclause:sat-by-positive} or \eqref{satclause:sat-by-negative}) yields 
\begin{equation}
\label{eq:refl-up-bound-base-1}
\neg I(j,m) \lor \neg T(m,\ell,b) \lor \left(D(1,j,\ell,b) \land T(\ell)^b\right). 
\end{equation}
Cutting \eqref{eq:refl-up-bound-base-1} for each $\ell \in [n]$ and $b \in \{0,1\}$ with \eqref{satclause:at-least-one-lit-sat} gives $\neg I(j,m) \lor D_{1,j}$. Cutting these clauses for $m \in [r]$ with \eqref{refclause:I-dom} yields $D_{1,j}$.

Induction step: Assume we have derived $D_{i-1,j'}$ for all $j' \in [t]$. For each $j \in [t]$ we shall derive $D_{i,j}$. Write $P_1$ in place of $L$ and  $P_0$ in place of $R$. 

For each $\ell \in [n], b \in \{0,1\}, j' \in [t]$, cut $\neg D(i-1,j',\ell,1) \lor \neg D(i-1,j',\ell,0)$ (from \eqref{refclause:nontaut}) with $\neg P_{1-b}(i,j,j') \lor \neg V(i,j,\ell) \lor D(i-1,j',\ell,1-b)$ (which is from \eqref{refclause:res-L-cut} or \eqref{refclause:res-R-cut}) to obtain $\neg P_{1-b}(i,j,j') \lor \neg V(i,j,\ell) \lor \neg D(i-1,j',\ell,b)$. Cut this with $D_{i-1,j'}$ to get
\begin{equation}
\label{eq:refl-up-bound-induc-1}
\neg P_{1-b}(i,j,j') \lor \neg V(i,j,\ell) \lor \left(D_{i-1,j'} \setminus \{ D(i-1,j',\ell,b) \land T(\ell)^b \}\right).
\end{equation}
Cutting \eqref{eq:refl-up-bound-induc-1} with axiom $T(\ell) \lor \neg T(\ell)$ yields
\begin{align}
 \begin{split}
\label{eq:refl-up-bound-induc-2}
\neg P_{1-b}(i,j,j') & \lor \neg V(i,j,\ell)\lor T(\ell)^{1-b}\\
& \lor \left( D_{i-1,j'} \setminus \{ D(i-1,j',\ell,0) \land \neg T(\ell), D(i-1,j',\ell,1) \land T(\ell) \} \right).
 \end{split}
\end{align}
Next, for each $\ell' \in [n] \setminus \{\ell\}$ and $b' \in \{0,1\}$, apply $\land$-introduction to  $T(\ell') \lor \neg T(\ell')$ and $\neg P_{1-b}(i,j,j') \lor \neg V(i,j,\ell) \lor \neg D(i-1,j',\ell',b') \lor D(i,j,\ell',b')$ (from \eqref{refclause:res-L-transf} or \eqref{refclause:res-R-transf}) to get
\begin{align}
 \begin{split}
\label{eq:refl-up-bound-induc-3}
\neg P_{1-b}(i,j,j') \lor \neg V(i,j,\ell) & \lor \left( D(i,j,\ell',b') \land T(\ell')^{b'} \right) \\
 & \lor \neg D(i-1,j',\ell',b') \lor T(\ell')^{1-b'}.
 \end{split}
\end{align}
Cutting \eqref{eq:refl-up-bound-induc-3}, for each $\ell'\in [n] \setminus \{\ell\}$ and $b' \in \{0,1\}$, with \eqref{eq:refl-up-bound-induc-2} results, after a weakening, in
\begin{equation}
\label{eq:refl-up-bound-induc-4}
\neg P_{1-b}(i,j,j') \lor \neg V(i,j,\ell) \lor T(\ell)^{1-b} \lor D_{i,j}.
\end{equation}
Recall that we have obtained \eqref{eq:refl-up-bound-induc-4} for each $\ell \in [n], b \in \{0,1\}, j' \in [t]$. For each $\ell \in [n]$ and $b \in \{0,1\}$, cut the clauses \eqref{eq:refl-up-bound-induc-4}, $j' \in [t]$, with $\bigvee_{j' \in [t]} P_{1-b}(i,j,j')$ (from \eqref{refclause:L-dom} or \eqref{refclause:R-dom}) to derive 
\begin{equation}
\label{eq:refl-up-bound-induc-5}
 \neg V(i,j,\ell) \lor T(\ell)^{1-b} \lor D_{i,j}.
\end{equation}
For each $\ell \in [n]$, cut \ref{eq:refl-up-bound-induc-5} for $b = 0$ and $b = 1$ on variable $T(\ell)$ to get $\neg V(i,j,\ell) \lor D_{i,j}$, and from these clauses derive $D_{i,j}$ by cuts with \eqref{refclause:V-dom}.

As for bounding the size of the refutation, the size of the base case is $O(trn^2+ tr^2)$, the total size of the induction steps is $O(st^2n^3 + st^3n)$, and the size of the finish is $O(n^2)$.
\end{proof}

\paragraph{Acknowledgement.} I thank Albert Atserias, Ilario Bonacina, Tuomas Hakoniemi and Moritz M\"uller for their comments.

\bibliography{mybiblio}

\begin{thebibliography}{10}

\bibitem{Atserias-Bonet2004}
Albert Atserias and Mar\'ia~Luisa Bonet.
\newblock On the automatizability of resolution and related propositional proof
  systems.
\newblock {\em Information and Computation}, 189(2):182--201, 2004.

\bibitem{atserias-muller2019-ref}
Albert Atserias and Moritz M{\"u}ller.
\newblock Automating resolution is {NP}-hard.
\newblock {\em arXiv e-prints}, Apr 2019.
\newblock \href {http://arxiv.org/abs/1904.02991v1}
  {\path{arXiv:1904.02991v1}}.

\bibitem{Bonet-Pitassi-Raz}
Mar\'ia~Luisa Bonet, Toniann Pitassi, and Ran Raz.
\newblock On interpolation and automatization for {F}rege systems.
\newblock {\em SIAM J. Comput.}, 29(6):1939--1967, 2000.
\newblock \href {http://dx.doi.org/10.1137/S0097539798353230}
  {\path{doi:10.1137/S0097539798353230}}.

\bibitem{cook-reckhow1979}
Stephen~A. Cook and Robert~A. Reckhow.
\newblock The relative efficiency of propositional proof systems.
\newblock {\em The Journal of Symbolic Logic}, 44(1):36--50, 1979.

\bibitem{Haken1985}
Armin Haken.
\newblock The intractability of resolution.
\newblock {\em Theoretical Computer Science}, 39(2--3):297--308, 1985.

\bibitem{Krajicek-Skelley-Thapen}
Jan Kraj\'{\i}\v{c}ek, Alan Skelley, and Neil Thapen.
\newblock {NP} search problems in low fragments of bounded arithmetic.
\newblock {\em The Journal of Symbolic Logic}, 72(2):649 -- 672, 2007.

\bibitem{Pudlak2000}
Pavel Pudl\'ak.
\newblock Proofs as games.
\newblock {\em American Mathematical Monthly}, 107(6):541--550, 2000.
\newblock \href {http://dx.doi.org/10.2307/2589349}
  {\path{doi:10.2307/2589349}}.

\bibitem{Pudlak2003}
Pavel Pudl\'ak.
\newblock On reducibility and symmetry of disjoint {NP}-pairs.
\newblock {\em Theoretical Computer Science}, 295:323--339, 2003.

\bibitem{segerlind-buss-impagliazzo}
Nathan Segerlind, Samuel~R. Buss, and Russel Impagliazzo.
\newblock A switching lemma for small restrictions and lower bounds for k-{DNF}
  resolution.
\newblock {\em SIAM Journal on Computing}, 33(5):1171--1200, 2004.

\bibitem{Thapen2016}
Neil Thapen.
\newblock A tradeoff between length and width in resolution.
\newblock {\em Theory of Computing}, 12(5):1--14, 2016.

\end{thebibliography}

\appendix

\section{Formula REF of Atserias and M\"uller}
\label{sec:REF-AM}
The purpose of this section is to give an answer to the lower bound question from \cite{atserias-muller2019-ref} in its original formulation, in which the refutation statement is formulated a bit differently from our $\textnormal{REF}^{F}_{s,t}$. 

We list the clauses of the formula $\textnormal{REF}(F,\tilde{s})$ of \cite{atserias-muller2019-ref}:
\begin{align}
& V[u,0] \lor V[u,1] \lor \ldots \lor V[u,n] & u \! \in \! [\tilde{s}], \\
& I[u,0] \lor I[u,1] \lor \ldots \lor I[u,r] & u \! \in \! [\tilde{s}], \\
& L[u,0] \lor L[u,1] \lor \ldots \lor L[u,\tilde{s}] & u \! \in \! [\tilde{s}], \\
& R[u,0] \lor R[u,1] \lor \ldots \lor R[u,\tilde{s}] & u \! \in \! [\tilde{s}], \\
& \neg V[u,i] \lor \neg V[u,i'] & u \! \in \! [\tilde{s}], i,i' \! \in \! [n] \cup \{0\}, i \neq i', \\
& \neg I[u,j] \lor \neg I[u,j'] & u \! \in \! [\tilde{s}], j,j' \! \in \! [r] \cup \{0\}, j \neq j', \\
& \neg L[u,v] \lor \neg L[u,v'] & u \! \in \! [\tilde{s}], v,v' \! \in \! [\tilde{s}] \cup \{0\}, v \neq v', \\
& \neg R[u,v] \lor \neg R[u,v'] & u \! \in \! [\tilde{s}], v,v' \! \in \! [\tilde{s}] \cup \{0\}, v \neq v', \\
\label{refclauseAM:switch-negative}& \neg I[u,0] \lor \neg V[u,0] & u \! \in \! [\tilde{s}],\\
\label{refclauseAM:switch-positive}& I[u,0] \lor V[u,0] & u \! \in \! [\tilde{s}],\\
& \neg I[u,0] \lor \neg L[u,0] & u \! \in \! [\tilde{s}],\\
& \neg I[u,0] \lor \neg R[u,0] & u \! \in \! [\tilde{s}],\\
& \neg L[u,v] & u,v \! \in \! [\tilde{s}], u \leq v,\\
& \neg R[u,v] & u,v \! \in \! [\tilde{s}], u \leq v,\\
& \neg L[u,v] \lor \neg V[u,i] \lor D[v,i,1] & u,v \! \in \! [\tilde{s}], i \! \in \! [n], b \! \in \! \{0,1\},\\
& \neg R[u,v] \lor \neg V[u,i] \lor D[v,i,0] & u,v \! \in \! [\tilde{s}], i \! \in \! [n], b \! \in \! \{0,1\},\\
& \neg L[u,v] \lor \neg V[u,i] \lor \neg D[v,i',b] \lor D[u,i',b] & u,v \! \in \! [\tilde{s}], i,i' \! \in \! [n], b \! \in \! \{0,1\}, i \neq i',\\
& \neg R[u,v] \lor \neg V[u,i] \lor \neg D[v,i',b] \lor D[u,i',b] & u,v \! \in \! [\tilde{s}], i,i' \! \in \! [n], b \! \in \! \{0,1\}, i \neq i',\\
& \neg I[u,j] \lor D[u,i,b] & u \! \in \! [\tilde{s}], j \! \in \! [r], x^b_i \! \in \! C_j,\\
& \neg D[u,i,0] \lor \neg D[u,i,1] & u \! \in \! [\tilde{s}], i \! \in \! [n],\\
& \neg D[\tilde{s},i,b] & i \! \in \! [n], b \! \in \! \{0,1\}.
\end{align}

The meanings of the variables and clauses of $\textnormal{REF}(F,\tilde{s})$ are very similar to those of $\textnormal{REF}^{F}_{s,t}$, which we described in words in detail, so let us concentrate on the main differences. First of all, the clauses described by $\textnormal{REF}(F,\tilde{s})$ through $D$-variables are indexed from 1 to $\tilde{s}$; this is their order in the refutation they should form (and they are not arranged in levels). Moreover, unlike in $\textnormal{REF}^{F}_{s,t}$ where each clause described by $D$-variables, with the exception of clauses on level 1, has to be derived only by the resolution rule, in $\textnormal{REF}(F,\tilde{s})$ there are both options (weakening of a clause in $F$ and the resolution rule). That exactly one of these options must be chosen in a valid resolution refutation is ensured with the help of \eqref{refclauseAM:switch-negative}, \eqref{refclauseAM:switch-positive}, and the additional value 0 that the second index of $V,I,L,R$-variables can attain. In particular, any assignment satisfying $\textnormal{REF}(F,\tilde{s})$ evaluates to 1 exactly one of the variables in $\{I[u,j]: j \in [r]\} \cup \{V[u,i]: i \in [n]\}$.

We show how a lower bound on the length of resolution refutations for the formula $\textnormal{REF}(F,\tilde{s})$ follows from the lower bound for $\textnormal{REF}^{F}_{n+1,t}$. 

Let $F$ be a CNF in $n$ variables with $r$ clauses and assume the parameter $\tilde{s}$ in $\textnormal{REF}(F,\tilde{s})$ is such that $t = \lfloor \frac{\tilde{s}}{n+1} \rfloor$ satisfies condition \eqref{eqn:conditions-on-nrst} of Theorem \ref{thm:main-lower-bound-REF-F-s-t} for $\textnormal{REF}^{F}_{n+1,t}$. 

It is straightforward to assign some variables of $\textnormal{REF}(F,\tilde{s})$ so that after removing the satisfied clauses, the formula becomes $\textnormal{REF}^{F}_{n+1,t}$ up to a renaming of variables (and after removing certain clauses from \eqref{refclause:res-L-transf} and \eqref{refclause:res-R-transf} in $\textnormal{REF}^{F}_{n+1,t}$, which immediately follow from \eqref{refclause:nontaut} - \eqref{refclause:res-R-cut} anyway). First, set the appropriate variables in $\textnormal{REF}(F,\tilde{s})$ so that the clauses $D_1, D_2, \ldots , D_{s - t(n+1)}$ described by the formula are all obtained, say, by a weakening of the clause $C_1 \in F$, and that none of these clauses is used as a premise of the resolution rule. Then, arrange the remaining clauses into $n+1$ levels of $t$ clauses. Evaluate to 1 all variables $L[u,0]$, $R[u,0]$ with $u$ on the first level, and evaluate to 0 all the remaining $L[u,v]$, $R[u,v]$ except for those with $u$ on the first highest level than $v$. Further, evaluate to 1 all variables $V[u,0]$ with $u$ on the first level, and evaluate to 0 all the remaining $V[u,i]$ except for those with $u$ from second to last level and a non-zero $i$. Next, evaluate to 1 all $I[u,0]$ with $u$ from second to last level, and evaluate to 0 all the remaining $I[u,j]$ variables except for those on the first level with a non-zero $j$. Finally, replace all the non-evaluated variables by the corresponding variables of $\textnormal{REF}^{F}_{n+1,t}$, respecting the above chosen arrangement to $n+1$ levels of $t$ clauses.

Since the substitution just described takes any refutation of $\textnormal{REF}(F,\tilde{s})$ to a refutation of $\textnormal{REF}^{F}_{n+1,t}$ without any increase in size, Theorem \ref{thm:main-lower-bound-REF-F-s-t}
implies an exponential resolution size lower bound for the original formula $\textnormal{REF}(F,\tilde{s})$ of Atserias and M\"uller, as stated in the following theorem. 

\begin{theorem} \label{thm:main-lower-bound-REF-AM}
For each $\epsilon>0$ there is $\delta>0$ and an integer $t_0$ such that if $n,r,\tilde{s}$ are integers satisfying 
\begin{equation*}
 r \geq n \geq 2,\qquad \left\lfloor \frac{\tilde{s}}{n+1} \right\rfloor \geq r^{3+\epsilon}, \qquad  \left\lfloor \frac{\tilde{s}}{n+1} \right\rfloor \geq t_0,
\end{equation*}
and $F$ is an unsatisfiable CNF consisting of $r$ clauses $C_1, \ldots , C_r$ in $n$ variables $x_1 , \ldots , x_n$, then any resolution refutation of $\textnormal{REF}(F,\tilde{s})$ has length greater than $2^{\lfloor \frac{\tilde{s}}{n+1} \rfloor ^{\delta}}$. 
\end{theorem}

\end{document}